\documentclass[a4paper,twocolumn,11pt,accepted=2024-10-09]{quantumarticle}
\pdfoutput=1
\usepackage[utf8]{inputenc}
\usepackage[english]{babel}
\usepackage[T1]{fontenc}
\usepackage{amsmath}
\usepackage{hyperref}
\usepackage{braket}
\usepackage{amsthm}
\usepackage{mathtools}
\usepackage{physics}
\usepackage{xcolor}
\usepackage{graphicx}
\usepackage[left=23mm,right=13mm,top=35mm,columnsep=15pt]{geometry} 
\usepackage{adjustbox}
\usepackage{placeins}
\usepackage[T1]{fontenc}
\usepackage{lipsum}
\usepackage{csquotes}
\usepackage{amsmath,amssymb}
\usepackage{graphicx}
\usepackage{tabularx,booktabs}
\usepackage{tikz}
\usepackage{float}
\usetikzlibrary{quantikz2}
\usepackage[utf8]{inputenc}
\usepackage[ruled,lined]{algorithm2e}
\usepackage{algpseudocode}
\usepackage{hyperref}

\usepackage{caption}

\newtheorem{theorem}{Theorem}
\newtheorem{proposition}{Proposition}
\newtheorem{lemma}{Lemma}
\newtheorem{corollary}{Corollary}

\theoremstyle{definition}
\newtheorem{definition}[theorem]{Definition}

\begin{document}

\title{Random Natural Gradient}

\author{Ioannis Kolotouros}
\affiliation{University of Edinburgh, School of Informatics, EH8 9AB Edinburgh, United Kingdom}
\email{i.kolotouros@sms.ed.ac.uk}
\author{Petros Wallden}
\affiliation{University of Edinburgh, School of Informatics, EH8 9AB Edinburgh, United Kingdom}
\email{petros.wallden@ed.ac.uk}
\maketitle

\begin{abstract}
Hybrid quantum-classical algorithms appear to be the most promising approach for near-term quantum applications. An important bottleneck is the classical optimization loop, where the multiple local minima and the emergence of barren plateaux make these approaches less appealing. To improve the optimization, the Quantum Natural Gradient (QNG) method [Quantum 4, 269 (2020)] was introduced -- a method that uses information about the local geometry of the quantum state-space. While the QNG-based optimization is promising, in each step it requires more quantum resources, since to compute the QNG one requires $O(m^2)$ quantum state preparations, where $m$ is the number of parameters in the parameterized circuit. In this work we propose two methods that reduce the resources/state preparations required for QNG, while keeping the advantages and performance of the QNG-based optimization. Specifically, we first introduce the Random Natural Gradient (RNG) that uses random measurements and the classical Fisher information matrix (as opposed to the quantum Fisher information used in QNG). The essential quantum resources reduce to linear $O(m)$ and thus offer a quadratic ``speed-up’’, while in our numerical simulations it matches QNG in terms of accuracy. We give some theoretical arguments for RNG and then benchmark the method with the QNG on both classical and quantum problems. Secondly, inspired by stochastic-coordinate methods, we propose a novel approximation to the QNG which we call Stochastic-Coordinate Quantum Natural Gradient that optimizes only a small (randomly sampled) fraction of the total parameters at each iteration. This method also performs well in our benchmarks, while it uses fewer resources than the QNG.
\end{abstract}

\section{Introduction}

We are approaching the era where quantum computers of $\approx 1000$ qubits will become widely available. Despite the increase in scale, these quantum devices still inherit imperfect operations and short coherence times making them unsuitable for certain quantum algorithms. To address these issues, people employed classical computers to work in conjunction with these imperfect devices and developed variational quantum algorithms (VQAs) \cite{cerezo2021variational, RevModPhys.94.015004}. However, the question of whether these hybrid quantum-classical architectures will allow for a practical advantage is still open.

In these approaches, the computational task of interest is transformed into the ground state of an interacting qubit Hamiltonian. The user then selects a properly suited parameterized quantum circuit, either problem-specific such as QAOA \cite{farhi2014quantumapproximateoptimizationalgorithm} or problem agnostic suited for the available quantum hardware \cite{kandala2017hardware}, and iteratively prepares and measures quantum states. Additionally, the classical computer post-processes the quantum measurements to calculate an objective function (or its higher-order derivatives) in order to update the parameters of the quantum circuit towards a descent direction. When the optimization terminates, the quantum computer returns a quantum state which is a solution (or an approximation) for the computational task considered. 

In order to make such a framework practical, certain conditions must be met. In the NISQ regime, we need to distinguish the classical from the quantum resources required to solve a problem. For the number of parameters used and their corresponding matrices (e.g. the Hessian) of size $m\times m$, the classical computers are powerful enough to perform standard linear algebra calculations with perfect accuracy while the quantum computers are still imperfect, with slow compilation times. For this reason, the number of quantum resources needed to solve a problem must be limited enough so that we still exploit the quantum effects, and the classical resources should be maximized as they offer speed and reliability.

In the past few years, a lot of interest has been focused on all subsequent parts of VQAs. People have utilized both gradient-based methods that exploit parameter-shift rules \cite{PhysRevA.103.012405, Schuld_2019, wierichs2022general} or gradient-free methods \cite{anand2021natural, zhao2020natural} that treat the quantum circuit as a black box. On the other hand, \cite{koczor2022quantum} proposed that one should construct a quadratic model of the loss landscape by performing measurements on the quantum computer and then minimize this quadratic approximation on the classical computer, reducing the overall quantum resources. Additionally, \cite{abbas2023quantum,bowles2023backpropagation} argued on whether we can construct algorithms 
that will allow VQAs to be trained as efficiently as classical neural networks by reducing the quantum overhead required for the gradient calculation. In this paper, we focus on the classical optimization part and especially on the information that the classical computer receives in order to update the parameters of the quantum circuit.

Most of the classical optimization algorithms (e.g. Gradient Descent, SPSA, COBYLA) treat the quantum circuit as a black box which outputs expectation values (or their first and high-order derivatives) without considering information about the underlying quantum state. Because of the non-convexity of the energy landscapes, the majority of these algorithms converge to sub-optimal local minima, which are numerous even in low-depth quantum circuits \cite{you2021exponentially, anschuetz2021critical}. Moreover, these algorithms can be rather costly, requiring constant communication between the quantum and classical resources until convergence. To address these issues, people have developed \emph{information-theoretic} methods such as the quantum natural gradient \cite{stokes2020quantum}, where at each step, the updates are performed based on local information of the state space. These methods are closely connected to imaginary time evolution \cite{mcardle2019variational, gacon2023variational}. 

In \cite{stokes2020quantum}, the authors generalized the idea of the \emph{natural gradient} \cite{amari1998natural} in the quantum setting and introduced a novel classical optimization algorithm that takes into consideration how small changes in the parameter space affect the generated quantum states. However, such updates require the calculation of the quantum Fisher information matrix (QFIM) \cite{meyer2021fisher} at each step, which in general is computationally expensive to calculate, and thus, using it in VQAs becomes impractical. The QFIM has been extensively used in the NISQ era either as a capacity measure \cite{haug2021capacity}, a generalization measure for quantum machine learning models \cite{haug2023generalization}, or even as a tool to construct naturally parameterized architectures \cite{haug2022natural}. The QNG has also been further extended in the case of noisy and nonunitary circuits \cite{koczor2022qngnonunitary}.

In this paper, we give two new optimization methods, improving on efficiency over the QNG optimizer. We first introduce a method that we call \emph{random natural gradient}. As we will show, preparing a quantum state and measuring it on a random basis (by applying a random unitary and measuring it on the computational basis) offers a significant speedup in a VQA optimization framework. Random measurements have previously been used to construct the classical shadows of quantum states \cite{huang2020predicting, sack2022avoiding, elben2023randomized}, to extend the size of a quantum computation beyond the physical qubits of a device \cite{lowe2023fast} or even to experimentally approximate the quantum Fisher information \cite{yu2021experimental}.

Then, inspired by classical coordinate-descent methods \cite{wright2015coordinate, nesterov2012efficiency, tseng2001convergence}, we propose an approximation to the quantum natural gradient that requires only a fraction of the total resources at every iteration that we call stochastic-coordinate quantum natural gradient.\\

\noindent\emph{Our Contributions:}
\begin{itemize}
    \item We introduce a novel optimization technique which we call \emph{random natural gradient} that is quadratically faster than the \emph{quantum natural gradient} and achieves significant speedup over classical optimization algorithms used. 

    \item We analyze how different measurements on parameterized quantum states can approximate the underlying geometry in the state space.

    \item We introduce conditions under which preconditioning the gradient of the loss with an information matrix will result in a descent direction.

    \item We introduce a novel approximation to the quantum natural gradient called \emph{stochastic-coordinate quantum natural gradient} that utilizes only a portion of the total parameters of the parameterized quantum circuit and benchmark it against the quantum natural gradient.

\end{itemize}

\noindent\emph{Structure:} In Sec. \ref{sec:preliminaries} we give the essential background on Variational Quantum Algorithms, on distance measures in the state-space and probability distribution spaces, and on the quantum natural gradient.  In Sec. \ref{sec:random_natural_gradient} we introduce a novel optimization algorithm called \emph{random natural gradient} that achieves a quadratic speed-up over the quantum natural gradient and performs optimally in practice, where we postpone explanations of why this method works for the following two sections. In Sec. \ref{sec:local_optimization} we discuss the concept of local optimization, and provide conditions under which updates will result in a descent direction. In Sec. \ref{sec:optimal_measuremnt} we explain how we can find measurement operators under which we gain the maximum information of the state space. In Sec. \ref{sec:stochastic_coordinate} we introduce a second optimization algorithm called \emph{stochastic-coordinate quantum natural gradient} that approximates the quantum natural gradient by utilizing only a small part of the total parameters of the quantum circuit. In Sec. \ref{sec:method_evaluation} we discuss the different problems used to benchmark our proposals 
in our experiments and the different metrics used to evaluate our algorithms. In \ref{sec:results} we benchmark our proposed algorithms with other classical optimization algorithms. We conclude in Sec. \ref{sec:discussion} with a general discussion of our results and future work.

\section{Preliminaries}
\label{sec:preliminaries}

\subsection{Variational Quantum Algorithms}
\label{sec:variational_quantum_algorithms}

Variational Quantum Algorithms (VQAs) refer to a class of hybrid quantum/classical algorithms where a quantum computer works in parallel with a classical computer employed with a classical optimization algorithm. This framework offers a practical framework in the NISQ setting but lacks generic theoretical guarantees about its performance.

Consider a mathematical problem that is mapped into a qubit-Hamiltonian consisting of $L = \mathcal{O}(\text{poly}(n))$ Pauli strings, where $n$ is the system size (i.e. the number of qubits). This Hamiltonian is chosen in a way so that its ground state corresponds to the solution of the initial problem. The most general way to write this Hamiltonian is:
\begin{equation}
    \mathcal{H} = \sum_{i=1}^L c_l P_l
\end{equation}
where $c_l \in \mathbb{R}$ is the real coefficient corresponding to Pauli string $P_l$.

As a first step in VQAs, the quantum computer utilizes a parameterized architecture $U(\boldsymbol{\theta})$, consisting of $m=\mathcal{O}(\text{poly}(n))$ parameterized gates along with an easy-to-prepare reference state (usually chosen to be the $\ket{0} \equiv \ket{0}^{\otimes n}$ state). Consider the most general unitary operator $U(\boldsymbol{\theta})$ parameterized by an $m$-dimensional vector $\boldsymbol{\theta} = (\theta_1, \ldots, \theta_m)$:
\begin{equation}
    U(\boldsymbol{\theta}) = \prod_{j=m}^1 e^{-i\theta_j g_j}
\label{eq:parameterized_quantum_circuit}
\end{equation}
where $g_j$ are the generators ($g_j^\dagger = g_j$) corresponding to each parameterized gate. Then, the quantum computer prepares a parameterized quantum state $\rho(\boldsymbol{\theta})$ and measures it (possibly in many different bases).

The classical computer is then used to post-process these measurements in order to compute the objective function. The objective function is usually the expectation value of the energy, but other choices have also been considered in the literature such as CVaR \cite{barkoutsos2020improving}, Ascending-CVaR \cite{kolotouros2022evolving} or Gibbs objective functions \cite{li2020quantum}. Moreover, the classical computer employs a classical optimization algorithm and calculates the direction to tune the parameters of the quantum architecture which points towards the direction of a (possibly local \cite{you2021exponentially, anschuetz2022quantum}) minimum.

The previous two steps are iteratively executed until the classical optimization has converged or other stopping criteria have been met. Finally, the VQA algorithm outputs both the optimal solution of the problem as well as the quantum state corresponding to the solution (or an approximation to them)
\begin{equation*}
    \rho(\boldsymbol{\theta^*}) , \;\; \mathcal{L}(\boldsymbol{\theta^*}).
\end{equation*}

\subsection{Distance of Probability Distributions}
\label{sec:distance_of_prob_distr}

As we discussed in Section \ref{sec:variational_quantum_algorithms} the quantum computer prepares a parameterized quantum state $\ket{\psi(\boldsymbol{\theta})} = U(\boldsymbol{\theta})\ket{0}$. Once the state has been prepared, a measurement basis is chosen, and the system of qubits is measured. The measurement basis can be changed by first applying a unitary matrix and then performing projective measurements on each qubit. 

Let $V$ be the unitary operator that changes the measurement basis of the system of qubits. We can assume that this unitary is parameterized by a $k = \mathcal{O}(\text{poly}(n))$-dimensional vector (i.e., it is comprised of a series of parameterized gates) $\boldsymbol{\phi} = (\phi_1, \phi_2, \ldots, \phi_k)$. Measurements on a basis that would require a super-polynomial unitary to be made to a local measurement are not practical (i.e. cannot be efficiently performed), and we will ignore them in our analysis. Overall, we can write the action of the ansatz family and the change of basis as:
\begin{equation}
    \ket{\psi(\boldsymbol{\theta}, \boldsymbol{\phi})} = V(\boldsymbol{\phi})U(\boldsymbol{\theta})\ket{0}.
\end{equation}

Once the measurement basis is selected, the system of qubits is prepared and measured a constant number of times with respect to a measurement basis $\mathcal{M} = \{ \Pi_l \}$. As a result, a different choice of basis $\mathcal{M}$ gives rise to a different probability distribution $p_{\mathcal{M}}(\boldsymbol{\theta}) = \{p_i(\boldsymbol{\theta})\}$ with $p_{\mathcal{M}}(\boldsymbol{\theta}) \succcurlyeq 0$, $\norm{p_{\mathcal{M}}(\boldsymbol{\theta})}_1 = 1$. The number of different probability outcomes for a $n$-qubit state is upper bounded by $2^n$ and depends on the measurement basis $\mathcal{M}$. However, if a quantum state is measured only $K$ times (with $K\ll 2^n$) then this number is bounded by $K$. For now, we will assume that the number of different measurement outcomes is $K$ and so $p_{\mathcal{M}}(\boldsymbol{\theta}) \in \Delta^{K-1}$, where $\Delta^{K-1}$ is the probability simplex of dimension $K-1$. The probability $p_l$ of each outcome is given as:
\begin{equation}
    p_l = \tr(V(\boldsymbol{\phi}) \rho(\boldsymbol{\theta}) V^\dagger(\boldsymbol{\phi})  \Pi_l)
\end{equation}
where $\Pi_l = \ketbra{l}$ is the projection operator on the $l$-th eigenspace. 

It will be very useful to introduce a measure that quantifies distances in the space of probability distributions. Let $\boldsymbol{u}, \boldsymbol{v} \in \Delta^{K-1}$ with $\norm{\boldsymbol{u}}_1=\norm{\boldsymbol{v}}_1 = 1$ be two probability distributions. The (\emph{Kullback-Leibler}) KL-divergence (or else the relative entropy) is defined as:
\begin{equation}
    \mathrm{KL}(\boldsymbol{u}||\boldsymbol{v}) = \sum_{j=1}^K u_j \log \frac{u_j}{v_j}.
\end{equation}
 The KL-divergence is not a metric since it is not symmetric under the interchange of $\boldsymbol{u}$ and $\boldsymbol{v}$ but satisfies all the properties of a monotonic distance measure. Specifically, the KL-divergence satisfies:
\begin{itemize}
    \item $\mathrm{KL}(\boldsymbol{u}||\boldsymbol{v}) = 0 \implies \boldsymbol{u}=\boldsymbol{v}$.

    \item $\mathrm{KL}(\boldsymbol{u}||\boldsymbol{v}) \geq 0$ for all $\boldsymbol{u},\boldsymbol{v}\in \Delta^{K-1}$.

    \item $\mathrm{KL}(T(\boldsymbol{u})||T(\boldsymbol{v})) \leq  \mathrm{KL}(\boldsymbol{u}||\boldsymbol{v})$ for every stochastic map $T$.
\end{itemize}

\subsubsection{Classical Fisher Information Matrix}

In our analysis below, we will assume that the measurement basis $\mathcal{M}$ is fixed. Thus, the resulting probability distribution will only depend on the choice of parameters $\boldsymbol{\theta}$. Let $p_{\mathcal{M}}(\boldsymbol{\theta})$ be the probability distribution after measuring the state $\ket{\psi(\boldsymbol{\theta})}$ and $p_{\mathcal{M}}(\boldsymbol{\theta+\boldsymbol{\epsilon}})$ be the probability distribution after measuring the state $\ket{\psi(\boldsymbol{\theta+\boldsymbol{\epsilon}})}$. If the shift vector $\boldsymbol{\epsilon}$ is small, we can Taylor expand the KL-divergence as:
\begin{equation}
\begin{gathered}
    \mathrm{KL}(p_{\mathcal{M}}(\boldsymbol{\theta})|| p_{\mathcal{M}}(\boldsymbol{\theta+\boldsymbol{\epsilon}})) = \mathrm{KL}(p_{\mathcal{M}}(\boldsymbol{\theta})|| p_{\mathcal{M}}(\boldsymbol{\theta})) + \\
    \sum_{i=1}^m  \epsilon_i\frac{\partial \mathrm{KL}(p_{\mathcal{M}}(\boldsymbol{\theta})|| p_{\mathcal{M}}(\boldsymbol{\theta}+\boldsymbol{\epsilon}))}{\partial \epsilon_i}\Big{|}_{\boldsymbol{\epsilon} = 0} + \\
    \frac{1}{2}\sum_{i,j=1}^m  \epsilon_i \epsilon_j\frac{\partial^2 \mathrm{KL}(p_{\mathcal{M}}(\boldsymbol{\theta})|| p_{\mathcal{M}}(\boldsymbol{\theta}+\boldsymbol{\epsilon}))}{\partial \epsilon_i\partial \epsilon_j}\Big{|}_{\boldsymbol{\epsilon} = 0 } + \mathcal{O}(\norm{\boldsymbol{\epsilon}}_1^3)
\label{eq:fisher_derivation}
\end{gathered}
\end{equation}
where if we neglect third-order terms, we can write:
\begin{equation}
        \mathrm{KL}(p_{\mathcal{M}}(\boldsymbol{\theta})|| p_{\mathcal{M}}(\boldsymbol{\theta+\boldsymbol{\epsilon}})) \approx \frac{1}{2}\boldsymbol{\epsilon}^T [\mathcal{F}_C^{\mathcal{M}}(\boldsymbol{\theta})] \boldsymbol{\epsilon} = \frac{1}{2}\norm{\boldsymbol{\epsilon}}_{\mathcal{F}_C^{\mathcal{M}}}^2
\end{equation}
where the first term in Eq. \eqref{eq:fisher_derivation} is zero since $ \mathrm{KL}(p_{\mathcal{M}}(\boldsymbol{\theta})|| p_{\mathcal{M}}(\boldsymbol{\theta}))= 0$ and the second term is also zero since it corresponds to a minimum\footnote{The KL-divergence is non-negative in general, and zero at $\boldsymbol{\epsilon}=0$.}. We can show that, for the choice of KL-divergence as a distance measure, the elements of the CFIM can be calculated as:
\begin{equation}
    [\mathcal{F}_C^{\mathcal{M}}(\boldsymbol{\theta})]_{ij} = \sum_l \frac{1}{p_l(\boldsymbol{\theta})}\frac{\partial p_l(\boldsymbol{\theta})}{\partial \theta_i} \frac{\partial p_l(\boldsymbol{\theta})}{\partial \theta_j}.
\end{equation}
For completeness, we have added the CFIM derivation in Appendix \ref{sec:classical_fisher_derivation} but one can also find it in \cite{meyer2021fisher}. It is worth noting that if a different distance measure was chosen, then it would always return a constant multiple of the CFIM as long as the choice of the distance measure is \emph{monotonic} \cite{morozova1991markov}.

\subsubsection{Measuring the Classical Fisher Information Matrix}

Let $O$ be any observable (Hermitian operator $O^\dagger = O$). The parameter shift rule \cite{romero2018strategies, Schuld_2019, PhysRevA.103.012405, wierichs2022general} states that the derivatives of the expectation value of an observable $O$ can be calculated as a linear combination of the expectation value of the observable at two different parameter settings:
\begin{equation}
    \frac{\partial \langle O(\boldsymbol{\theta}) \rangle}{\partial \theta_j} = r \left[ \langle O(\boldsymbol{\theta} + \frac{\pi}{4r}\boldsymbol{\hat{e}_j}) \rangle - \langle O(\boldsymbol{\theta} - \frac{\pi}{4r}\boldsymbol{\hat{e}_j}) \rangle \right]
\end{equation}
where $\pm r$ are the eigenvalues of the generator $g_j$ (see Eq. \eqref{eq:parameterized_quantum_circuit}) corresponding to the gate of the parameter $\theta_j$\footnote{The interested reader can find general parameter-shift rules in \cite{wierichs2022general}.}. In our case, $O = \Pi_l = \ket{l}\bra{l}$. Then, the CFIM elements, given by Eq. \eqref{eq:classical_fisher_elements}, can be calculated as:
\begin{equation}
\small
\begin{gathered}
    [\mathcal{F}_C]_{ij} = \sum_l \frac{r^2}{ \langle \Pi_l(\boldsymbol{\theta}) \rangle}\left[ \langle \Pi_l(\boldsymbol{\theta} + \frac{\pi}{4r}\boldsymbol{\hat{e}_i}) \rangle - \langle \Pi_l(\boldsymbol{\theta} - \frac{\pi}{4r}\boldsymbol{\hat{e}_i}) \rangle \right] \\
    \cross \left[ \langle \Pi_l(\boldsymbol{\theta} + \frac{\pi}{4r}\boldsymbol{\hat{e}_j}) \rangle - \langle \Pi_l(\boldsymbol{\theta} - \frac{\pi}{4r}\boldsymbol{\hat{e}_j}) \rangle \right].
\label{eq:classical_fisher_elements}
\end{gathered}
\end{equation}

We can see that the elements of the CFIM can be expressed as products of first-order derivatives. As such, we can introduce Corollary \ref{corollary:cfim_resources} that quantifies the classical and quantum resources needed for the calculation of the CFIM.
\begin{corollary}
        Consider a parameterized quantum circuit consisting of $m$ parameterized quantum gates. Any classical Fisher information matrix (CFIM) requires $\mathcal{O}_Q(m)$ different quantum state preparations and $\mathcal{O}_C(m^2)$ classical resources to post-process the measurements and store the matrix.
\label{corollary:cfim_resources}
\end{corollary}
As we discuss later, this results in the CFIM requiring quadratically less quantum resources than the QFIM.

\subsection{Distance of pure density operators}

Just as we defined a measure of distance in the space of probability distributions, we could also measure distances in the space of density operators. In this paper, we will focus only on \emph{pure} quantum states ($\tr\rho^2 = 1$).

As we discussed for the classical case, there is a unique underlying metric in the space of probability distributions independent of the choice of distance measure. Different distance measures will always yield a constant multiple of the CFIM. However, this is not true in the quantum case as Petz \cite{petz1996monotone} proved that there exist infinitely many metrics. If we restrict ourselves to the space of \emph{pure} quantum states, there is a unique underlying metric independent of the choice of the distance measure \cite{stokes2020quantum}. As a distance measure, we choose the \emph{infidelity} between pure quantum states which for two quantum states $\ket{\psi(\boldsymbol{\theta})}$, $\ket{\psi(\boldsymbol{\theta} + \boldsymbol{\epsilon})})$ is defined as:
\begin{equation}
    d_f(\ket{\psi(\boldsymbol{\theta})}, \ket{\psi(\boldsymbol{\theta} + \boldsymbol{\epsilon})}) = 1 - |\bra{\psi(\boldsymbol{\theta})}\ket{\psi(\boldsymbol{\theta }+ \boldsymbol{\epsilon})}|^2.
\end{equation}

\subsubsection{Quantum Fisher Information Matrix}

Let $\ket{\psi(\boldsymbol{\theta})}$ and $\ket{\psi(\boldsymbol{\theta} + \boldsymbol{\epsilon})}$ be two parameterized quantum states. If the shift vector $\boldsymbol{\epsilon}$ is small, then we can Taylor expand the infidelity as:
\begin{equation*}
\begin{gathered}
    d_f(\ket{\psi(\boldsymbol{\theta})},\ket{\psi(\boldsymbol{\theta}+ \boldsymbol{\epsilon})}) = d_f(\ket{\psi(\boldsymbol{\theta})},\ket{\psi(\boldsymbol{\theta})}) + \\
    \sum_{i=1}^m \epsilon_i \frac{\partial  d_f(\ket{\psi(\boldsymbol{\theta})},\ket{\psi(\boldsymbol{\theta}+ \boldsymbol{\epsilon})})}{\partial \epsilon_i}\Bigg{|}_{\boldsymbol{\epsilon}=0} + \\
    \frac{1}{2}\sum_{i,j=1}^m \epsilon_i \epsilon_j\frac{\partial^2  d_f(\ket{\psi(\boldsymbol{\theta})},\ket{\psi(\boldsymbol{\theta}+ \boldsymbol{\epsilon})})}{\partial \epsilon_i \partial \epsilon_j} \Bigg{|}_{\boldsymbol{\epsilon}=0} +  \mathcal{O}(\norm{\boldsymbol{\epsilon}}_1^3) 
\end{gathered}
\end{equation*}
where if we neglect third-order terms, we can write
\begin{equation}
 d_f(\ket{\psi(\boldsymbol{\theta})},\ket{\psi(\boldsymbol{\theta}+ \boldsymbol{\epsilon})}) \approx \frac{1}{4} \boldsymbol{\epsilon}^T [\mathcal{F}_Q(\boldsymbol{\theta})] \boldsymbol{\epsilon}
    = \frac{1}{4} \norm{\boldsymbol{\epsilon}}_{\mathcal{F}_Q}^2 
\label{eq:quantum_fisher}
\end{equation}
where similarly to the classical relative entropy, the first two terms vanish because $ d_f(\ket{\psi(\boldsymbol{\theta})},\ket{\psi(\boldsymbol{\theta})})$ corresponds to a minimum with a value equal to zero. The matrix $\mathcal{F}_Q(\boldsymbol{\theta})$ is called the \emph{quantum Fisher information matrix} (QFIM) and acts as a metric in the space of quantum states, giving information about the geometry of states near the vicinity of the state $\ket{\psi(\boldsymbol{\theta})}$. The matrix elements of the QFIM are calculated to be the real part of the \emph{Fubini-Study metric} (see \cite{stokes2020quantum, meyer2021fisher}):
\begin{equation}
\begin{gathered}
    [\mathcal{F}_Q(\boldsymbol{\theta})]_{ij} = 4 \; \mathrm{Re}\Bigg[\frac{ \partial \bra{\psi(\boldsymbol{\theta})}}{\partial \theta_i} \frac{\partial \ket{\psi(\boldsymbol{\theta})}}{\partial \theta_j} \\
   - \frac{ \partial \bra{\psi(\boldsymbol{\theta})}}{\partial \theta_i}\ket{\psi(\boldsymbol{\theta})}\bra{\psi(\boldsymbol{\theta})}\frac{ \partial \ket{\psi(\boldsymbol{\theta})}}{\partial \theta_j} \Bigg].
\label{eq:qfim_elements}
\end{gathered}
\end{equation}

Intuitively, the QFIM acts as a metric in the space of quantum states. It provides a description of the geometry of the underlying state, giving information on how the parameterized quantum state changes if we vary a given parameter. Large eigenvalues of the QFIM will result in significant changes in the quantum state (with respect to a distance measure) even for small variations towards the direction of the corresponding eigenvector. On the other hand, zero eigenvalues will correspond to \emph{singularities}, i.e. points in the space of parameters where changes will have no effect on the underlying quantum state \cite{haug2021capacity}.

\subsubsection{Measuring the Quantum Fisher Information Matrix}
\label{subsec:measuring_qfisher}

So far in the literature, there has been extensive research on how to calculate the elements of QFIM given by Eq. \eqref{eq:qfim_elements}. In \cite{PhysRevA.103.012405} explained how to use parameter-shift rules to calculate QFIM while \cite{banchi2021measuring} introduced stochastic parameter-shift rules. On the other hand, \cite{haug2021capacity, mcardle2019variational} proposed an alternative way with the Hadamard-overlap test, using an extra ancilla qubit, but requires larger device connectivity. Overall, the quantum and classical resources needed for the calculation of the QFIM are stated in Corollary \ref{corollary:qfim_resources}.

\begin{corollary}
    Consider a parameterized quantum circuit, consisting of $m$ parameterized quantum gates. The quantum Fisher information matrix at any parameter configuration $\boldsymbol{\theta}$ requires $\mathcal{O}_Q(m^2)$ different quantum state preparations and $\mathcal{O}_C(m^2)$ classical resources to post-process the measurement and store the matrix.
\label{corollary:qfim_resources}
\end{corollary}

\subsection{Quantum Natural Gradient}

The Quantum Natural Gradient (QNG) \cite{stokes2020quantum} is an optimization algorithm suited for variational quantum algorithms. Instead of updating the parameters in the direction of the negative gradient (of the loss function), the algorithm considers the changes happening in the space of parameterized quantum states. Specifically, at each iteration, the parameters are changed according to the update rule:
\begin{equation}
    \boldsymbol{\theta_{k+1}} = \boldsymbol{\theta_k} -\eta \mathcal{F}_Q(\boldsymbol{\theta_k})^{+}\grad \mathcal{L}(\boldsymbol{\theta_k})
\label{eq:qng_update_rule}
\end{equation}
where $\mathcal{F}_Q(\boldsymbol{\theta_k})^{+}$ is the \emph{Moore-Penrose inverse} (pseudoinverse) \footnote{The \emph{pseudoinverse} of a matrix $A$ corresponds to the inverse of $A$ that is defined on the space orthogonal to the kernel.} of the \emph{quantum Fisher information matrix} (QFIM). Intuitively, QNG performs large steps in the directions where the state changes by a small amount and takes smaller steps in the directions where the state changes by a large amount. Although the simulated experiments in \cite{stokes2020quantum} showed that the convergence speed (in terms of optimization iterations) is improved significantly compared to first-order local optimizers, the bottleneck is that it requires $\mathcal{O}_Q(m^2)$ state preparations at each iteration which results in a big drawback for near term devices. This has the immediate implication that the actual quantum resources needed to implement the algorithm are quite large limiting its actual practicality. We discuss that thoroughly in the next sections.

\subsubsection{Steepest Descent}
\label{subsubsection:steepest_descent}

At this point, it would be fruitful to provide a more geometric explanation behind the update rule \eqref{eq:qng_update_rule} of QNG and its approximation through the update rule of Eq. \eqref{eq:classical_fisher_update}. Consider the first-order Taylor expansion of loss function $\mathcal{L}(\boldsymbol{\theta})$ at the point $\boldsymbol{\theta} + \boldsymbol{v}$:
\begin{equation}
    \mathcal{L}(\boldsymbol{\theta} + \boldsymbol{v}) \approx \mathcal{L}(\boldsymbol{\theta}) +\grad \mathcal{L}(\boldsymbol{\theta})^T \boldsymbol{v}.
\end{equation}
Steepest descent \cite{boyd2004convex} aims to find a direction $\boldsymbol{v}$ such that the directional derivative becomes as small as possible. Since $\grad \mathcal{L}(\boldsymbol{\theta})^T \boldsymbol{v}$ is linear in $\boldsymbol{v}$, it can be made as small as we desire. So, to make the question sensible, we define the \emph{normalized steepest descent} as:
\begin{equation}
    \Delta \boldsymbol{\theta}_{\text{nsd}} = \arg \min_{\boldsymbol{v}:\norm{\boldsymbol{v}}=1}\{\grad \mathcal{L}(\boldsymbol{\theta})^T \boldsymbol{v}\}
\label{eq:steepest_descent}
\end{equation}
where $\norm{\cdot}$ can be any vector norm. For example, choosing $\norm{\cdot} = \norm{\cdot}_2$ (the $l_2$-norm) the steepest descent becomes gradient descent or choosing $\norm{\cdot} = \norm{\cdot}_1$ (the $l_1$-norm) becomes \emph{coordinate descent} \cite{boyd2004convex}. We are interested in the case where $\norm{\cdot}$ is chosen to be $\norm{\cdot}_P$, where $P \succcurlyeq 0$ describes the intrinsic geometry in the parameterized space (or at least an approximation of it). Notice that the underlying geometry of the parameters of a parameterized quantum circuit is not \emph{Euclidean} but follows a \emph{Riemannian structure}. That is, the distance between vectors $\boldsymbol{\theta}$ and $\boldsymbol{\theta} + \delta \boldsymbol{\theta}$ is:
\begin{equation}
    d(\boldsymbol{\theta}, \boldsymbol{\theta} + \delta \boldsymbol{\theta})^2 = \delta\boldsymbol{\theta} G(\boldsymbol{\theta}) \delta \boldsymbol{\theta} = \norm{\delta\boldsymbol{\theta}}_{G(\boldsymbol{\theta})}^2
\end{equation}
where $G(\boldsymbol{\theta})$ is the \emph{Riemannian metric}. Since parameterized quantum circuits generate parameterized quantum states $\ket{\psi(\boldsymbol{\theta})}$, the actual geometry of the parameters is characterized by changes happening in the quantum state. As a result, the corresponding metric that describes the geometry of the parameters at a point $\boldsymbol{\theta}$ is the QFIM, i.e. $G(\boldsymbol{\theta}) = \mathcal{F}_Q(\boldsymbol{\theta})$.
\begin{lemma}
    Consider a parameterized quantum circuit parameterized by a vector $\boldsymbol{\theta}$ that generates a quantum state $\ket{\psi(\boldsymbol{\theta})}$. The normalized steepest descent direction of the loss function $\mathcal{L}(\boldsymbol{\theta})$ is:
    \begin{equation}
        - \frac{\mathcal{F}_Q^{-1}(\boldsymbol{\theta})\grad\mathcal{L}(\boldsymbol{\theta})}{\norm{\mathcal{F}_Q^{-1/2}(\boldsymbol{\theta})\grad\mathcal{L}(\boldsymbol{\theta})}_2}
    \end{equation}
    where $\mathcal{F}_Q(\boldsymbol{\theta})$ is the QFIM at point $\boldsymbol{\theta}$
\end{lemma}
\begin{proof}
    In the case of a parameterized quantum circuit, the local geometry of the parameterized space is described by the QFIM. As such, distances (near a point $\boldsymbol{\theta}$) are measured with respect to the $\norm{\cdot}_P$, with $P = \mathcal{F}_Q(\boldsymbol{\theta})$. Then, the normalized steepest descent direction is given by:
    \begin{equation}
    \Delta \boldsymbol{\theta}_{\text{nsd}} = \arg \min_{\boldsymbol{v}:\norm{\boldsymbol{v}}_{\mathcal{F}_Q}=1}\{\grad \mathcal{L}(\boldsymbol{\theta})^T \boldsymbol{v}\}.
    \label{eq:steep_pqc}
    \end{equation}
We can make use of the substitution $\boldsymbol{u} = \mathcal{F}_Q^{1/2} \boldsymbol{v}$. In this case, the constraint $\norm{\boldsymbol{v}}_{\mathcal{F}_Q}=1$ is replaced by:
\begin{equation}
    \norm{\boldsymbol{v}}_{\mathcal{F}_Q}=1 \implies \norm{\boldsymbol{u}}_2=1.
\end{equation}
As such, solving Eq. \eqref{eq:steep_pqc} is equivalent to:
    \begin{equation}
    \Delta \boldsymbol{\theta}_{\text{nsd}} = \arg \min_{\boldsymbol{u}:\norm{\boldsymbol{u}}_2=1}\{\grad \mathcal{L}(\boldsymbol{\theta})^T \mathcal{F}_Q^{-1/2}\boldsymbol{u}\}.
    \label{eq:steep_pqc_trans}
    \end{equation}
Thus, by solving Eq. \eqref{eq:steep_pqc_trans} and making again the substitution $v = \mathcal{F}_Q^{-1/2}u$ we find that the steepest descent direction is:
    \begin{equation*}
        - \frac{\mathcal{F}_Q^{-1}(\boldsymbol{\theta})\grad\mathcal{L}(\boldsymbol{\theta})}{\norm{\mathcal{F}_Q^{-1/2}(\boldsymbol{\theta})\grad\mathcal{L}(\boldsymbol{\theta})}_2}.
    \end{equation*}
\end{proof}

\section{Random Natural Gradient}
\label{sec:random_natural_gradient}

In this section, we will outline our first main result, which is a novel optimization algorithm called \emph{Random Natural Gradient} (RNG). The update rule of RNG is given by the formula:
\begin{equation}
\begin{gathered}
    \boldsymbol{\theta}_{k+1} = \boldsymbol{\theta}_k - \eta [\mathcal{F}_C^{\mathcal{M}}(\boldsymbol{\theta}_k)]^{+} \grad{\mathcal{L}}(\boldsymbol{\theta}_k) \\
    \text{Random measurement $\mathcal{M}$}
\label{update_rule_rng}
\end{gathered}
\end{equation}
where the random measurement $\mathcal{M}$ is performed by first applying a random unitary $V(\boldsymbol{\phi})$ on the parameterized state $\ket{\psi(\boldsymbol{\theta}_k)}$ and then measuring on the computational basis. In our analysis, we choose the random unitaries to be hardware-efficient parameterized quantum circuits for which we uniformly sample the parameters of the parameterized gates.

The intuition behind this update rule is relatively simple and is thoroughly explained in the following two sections. As we will explain, the optimization part in VQAs requires an update rule that can be calculated efficiently for any hope of quantum advantage. Additionally, the learning rate should carry some information about the corresponding state and how an infinitesimal change in a parameter would change it.

First of all, we can visualize measuring on a random basis and constructing the CFIM as a way to approximate the QFIM. Similarly, the authors of \cite{stokes2020quantum} suggested that instead of calculating the QFIM (as it would result in a tedious calculation), one could calculate only block-diagonal elements. This strategy would require fewer quantum resources but carries no physical intuition on why such an approximation would be valid, especially when there are high correlations between elements of different blocks.

Finding the measurement basis that would result in the optimal step is not a task that can be calculated efficiently. For that reason, one could draw a measurement basis at random and then construct the CFIM on that basis. Our experiments showed that a random CFIM has an increased rank (e.g. compared to measurements in $Z$-basis), and as such, more directions can be explored during the optimization (see Figure \ref{fig:ranks}). This has the further implication that the classical optimization part may avoid local minima as was explored in the QNG case in \cite{wierichs2020avoiding}. Similar findings were also found in \cite{garcia2023effects}, where the noise may increase the rank of QFIM, so more directions can be explored.

From a practical perspective, the update rule \eqref{update_rule_rng} can be calculated efficiently (see Corollary \ref{corollary:cfim_resources}) and, as we will show, improves the convergence dramatically. At each time step, the quantum resources (the number of quantum state preparations) needed are $2m$ for the gradient calculation and $2m+1$ for the calculation of the CFIM on a random basis (see Eq. \eqref{eq:classical_fisher_elements}) for a total of $4m + 1$ quantum states. Then, the classical computer post-processes the $2m + 1$ random basis measurements and utilizes a classical memory of size $m\times m$ for the random CFIM. This update rule is iteratively applied with a different measurement basis until convergence to a local (or a global) minimum occurs. 

Furthermore, our method inherits the advantage that the depth of the quantum circuit required to calculate the matrix elements of the random CFIM is less than that of the QFIM. In RNG, the depth of the additional unitary that is required for the calculation of the random CFIM is user-dependent and as we show in our experiments in Sec. \ref{sec:results}, the user can select shallow random quantum circuits and converge significantly faster than QNG.

Specifically, as we discussed in Sec \ref{subsec:measuring_qfisher}, for the calculation of QFIM, one requires quantum circuits of depth twice the one needed to generate the parameterized quantum state $\ket{\psi(\boldsymbol{\theta})}$. This imposes a bottleneck for parameterized quantum states that require large depths. However, for RNG, the additional quantum circuit is user-specified and depends on the architecture required for the random measurement (see more details in \ref{appendix:depth_overhead}). Our algorithm is outlined in Algorithm \ref{alg:random_natural_gradient}.

Finally, we would like to stress how RNG avoids singularities in the parameter space \cite{yamamoto2019natural}. Consider a fixed measurement basis $\mathcal{M}$. There are points in the parameters space where a small displacement in the parameters may not result in any change in the probabilities observed. This would result in a CFIM with degenerate zero eigenvalues. In a practical scenario, close to such a point, a natural gradient optimizer will make very large steps, prohibiting it from convergence. However, by switching the basis we can now avoid the singularities as for the new observables, the previous point may result in completely different probability distributions, and as such the optimizer will continue making small steps.

\begin{algorithm}[h!]
\caption{Random Natural Gradient}
\label{alg:random_natural_gradient}
\SetKwInOut{Input}{Input}
\Input{Problem Hamiltonian $\mathcal{H}$\;
Ansatz family $\ket{\psi(\boldsymbol{\theta})}=U(\boldsymbol{\theta})\ket{0}$\;
Total iterations $K$\;
Loss function $\mathcal{L}(\boldsymbol{\theta})$\;
Initial parameters $\boldsymbol{\theta} = \boldsymbol{\theta_0}$\;
Learning rate $\eta$\;}
\For{$k=1,2,\ldots, K$}{
Calculate derivatives $\frac{\partial \mathcal{L}(\boldsymbol{\theta})}{\partial \theta_i} \; \forall \; i \in [M]$\;
Shuffle a measurement basis $\mathcal{M}$\;
Calculate the Classical Fisher Information Matrix $\mathcal{F}_C^{\mathcal{M}}(\boldsymbol{\theta})$\;
Update $\boldsymbol{\theta}$ as $\boldsymbol{\theta} = \boldsymbol{\theta} - \eta [\mathcal{F}_C^{\mathcal{M}}(\boldsymbol{\theta})]^{+} \grad{\mathcal{L}}(\boldsymbol{\theta})$\;
}
\Return $\boldsymbol{\theta}$
\end{algorithm} 

\section{Local Optimization}
\label{sec:local_optimization}

In this section, we provide the motivation and theoretical intuition behind the update rule of RNG in Eq. \eqref{update_rule_rng} and argue why such an update is desirable. \emph{Local optimization} refers to the technique of following a trajectory in a region of the loss landscape and converging to a (possibly local) minimum. In this paper, we consider the loss function to be the expectation value of the energy of a parameterized quantum state:
\begin{equation}
    \mathcal{L}(\boldsymbol{\theta}) = \tr (\rho(\boldsymbol{\theta})\mathcal{H})
\label{eq:loss_function}
\end{equation}
where $\mathcal{H}$ is the Hamiltonian of the problem. Vanilla Gradient Descent (GD) iteratively updates the parameters $\boldsymbol{\theta}$ by following the direction of the negative gradient. The update rule is given by:
\begin{equation}
    \boldsymbol{\theta}_{k+1} = \boldsymbol{\theta}_k -\eta \grad \mathcal{L}(\theta_k)
\label{eq:gradient_descent}
\end{equation}
where $\eta > 0$ is a tunable hyperparameter that is crucial for both the speed and convergence of the algorithm. The biggest bottleneck in applying GD in a VQA setting is that a small choice of $\eta$ will require numerous quantum state preparations in the quantum computer, especially in the region where $\grad \mathcal{L}(\boldsymbol{\theta})\rightarrow 0$, while a large choice of $\eta$  may result in ``missing'' the minimum. In the former case, the quantum computer will require a very large number of circuit repetitions to acquire the desired accuracy, but also multiple and different quantum state preparations. A useful tool from the classical optimization literature is the \emph{proximal point method} \cite{parikh2014proximal} where the update rule is given by:
\begin{equation}
\begin{gathered}
    \boldsymbol{\theta_{k+1}} = \text{prox}_{\mathcal{L}, \lambda}(\boldsymbol{\theta_k}) \\
    = \arg \min_{\boldsymbol{\theta}}\left[ \mathcal{L}(\boldsymbol{\theta}) + \lambda q(\boldsymbol{\theta}, \boldsymbol{\theta_k})\right]
\label{eq:proximal_update}
\end{gathered}
\end{equation}
where $q(\boldsymbol{\theta}, \boldsymbol{\theta_k})$ is a \emph{dissimilarity function} that measures distance between the two vectors $\boldsymbol{\theta}, \boldsymbol{\theta_k}$. When $q$ is chosen to be the squared Euclidean distance:
\begin{equation}
    q(\boldsymbol{\theta}, \boldsymbol{\theta_k}) = \frac{1}{2}\norm{\boldsymbol{\theta}- \boldsymbol{\theta_k}}^2_2
\end{equation}
then the proximal update becomes the ordinary GD given by the update rule in Eq. \eqref{eq:gradient_descent} with $\eta=\lambda^{-1}$. In that case, the dissimilarity function acts as a penalty term that prohibits big steps in the space of parameters.

In \cite{stokes2020quantum} the authors claimed that the classical optimization algorithm should adapt the updated parameters according to the changes happening in the state-space i.e. update $\boldsymbol{\theta}$ according to:
\begin{equation}
    \boldsymbol{\theta}_{k+1} = \arg \min_{\boldsymbol{\theta}}\left[\mathcal{L}(\boldsymbol{\theta}) + \frac{1}{2\eta}\norm{\boldsymbol{\theta} - \boldsymbol{\theta_k}}^2_{\mathcal{F}_Q}\right]
\end{equation}
where the term $\norm{\boldsymbol{\theta} - \boldsymbol{\theta_k}}^2_{\mathcal{F}_Q}\equiv(\boldsymbol{\theta} - \boldsymbol{\theta_k})^T \mathcal{F}_Q (\boldsymbol{\theta} - \boldsymbol{\theta_k})$ penalizes large steps in the state-space \footnote{Here $\mathcal{F}_Q$ acts a metric, stretching the vectors in the state-space accordingly.}. In this case, the update rule \eqref{eq:proximal_update} is reformulated to the \emph{Quantum Natural Gradient} (QNG) where the parameters are iteratively changed according to the rule given by Eq. \eqref{eq:qng_update_rule}.

At this point, we would like to state that the QNG update rule in Eq. \eqref{eq:qng_update_rule} falls into the more general category of \textit{preconditioning}. In general, preconditioning the GD update in Eq. \eqref{eq:gradient_descent} by a positive definite matrix $A$:
\begin{equation}
\begin{gathered}
    \boldsymbol{\theta_{k+1}} = \boldsymbol{\theta_k} - A^{-1}\grad \mathcal{L}(\boldsymbol{\theta_k})\\
    A \succ 0
\end{gathered}
\end{equation}
results in a \emph{descent direction}. To see this, consider the Taylor expansion of the loss function \eqref{eq:loss_function} around the current iterate $\boldsymbol{\theta_k}$:
\begin{equation}
\begin{aligned}
    \mathcal{L}(\boldsymbol{\theta}) = \mathcal{L}(\boldsymbol{\theta_k}) + \grad\mathcal{L}(\boldsymbol{\theta_k})^T(\boldsymbol{\theta}-\boldsymbol{\theta_k}) \\ 
    +\frac{1}{2}(\boldsymbol{\theta}-\boldsymbol{\theta_k})^TH(\boldsymbol{\theta}-\boldsymbol{\theta_k}) + \mathcal{O}(\norm{\boldsymbol{\theta}-\boldsymbol{\theta_k}}^3)
\end{aligned}
\end{equation}
where $H$ is the Hessian at the point $\boldsymbol{\theta_k}$ and let the updated point be $\boldsymbol{\theta_k} - A^{-1}\grad \mathcal{L}(\boldsymbol{\theta_k})$. Keeping only the first-order terms (this would be valid for example if we are in a region with small gradients or if we scale the matrix $A$ by a small factor $\eta$) then the loss function is:

\begin{equation}
\small
\begin{gathered}
    \mathcal{L}(\boldsymbol{\theta_k}- A^{-1}\grad \mathcal{L}(\boldsymbol{\theta_k})) = \mathcal{L}(\boldsymbol{\theta_k}) - \grad\mathcal{L}(\boldsymbol{\theta_k})^T A^{-1} \grad\mathcal{L}(\boldsymbol{\theta_k})\\
    \implies \mathcal{L}(\boldsymbol{\theta_k}- A^{-1}\grad \mathcal{L}(\boldsymbol{\theta_k}))  <  \mathcal{L}(\boldsymbol{\theta_k})
\end{gathered}
\end{equation}
since the second term in the first line is negative for $A^{-1}\succ 0$. By keeping second-order terms, the condition so that the preconditioner $A$ points towards a descent direction becomes \footnote{We write $B \prec C$ to denote the matrix $B-C$ being negative definite.}:
\begin{equation}
    A^{-1}HA^{-1}\prec 2A^{-1}.
\label{eq:second-order-condition}
\end{equation}
The above analysis can readily be formulated in the case where the preconditioner is \emph{a positive semidefinite matrix}, but the inverse is replaced by the Moore-Penrose inverse. In this case, we observe two things. The first is that moving towards a descent direction is feasible when the matrix $A^{-1}$ is small (with respect to a matrix norm), which can always be done by multiplying by a sufficiently small scalar $\eta$. However, choosing a matrix that is computationally expensive to calculate and then scaling it by a small factor $\eta$ (see QNG update in Eq. \eqref{eq:qng_update_rule} and Corollary \ref{corollary:qfim_resources}) may prohibit any advantage of using the preconditioner in the first place. On the other hand, condition \eqref{eq:second-order-condition} filters a large amount of positive definite matrices that allow for a descent direction, but testing the condition in an online setting is impractical since it requires the calculation of the Hessian at every iterate. 

We argue here that the preconditioner should carry information about the intrinsic geometry of the parameters (just like in QNG) but at the same time, it should be relatively fast to calculate. As we propose in Section \ref{sec:random_natural_gradient}, a clever way to feed information about changes happening in the quantum state in a positive semidefinite matrix is to use random measurements. This alternative allows for a fast calculation of a positive semidefinite matrix which is intrinsically meaningful and improves the convergence in a VQA framework.

\subsection{Classical Natural Gradient}

In section Sec. \ref{sec:local_optimization}, we introduced the idea of proximal updates where the parameters are updated in a way that considers a distance measure of the parameters. Similarly to QNG, one could prepare a quantum state, measure it (e.g. on the computational basis), and, with the measurement outcomes, approximate the probability distribution of different outcomes. In that case, we can choose the dissimilarity function to be the KL-divergence (see Sec. \ref{sec:distance_of_prob_distr}) between the probability distributions after measuring the quantum states at the computational basis.

However, nothing prevents us from using a different dissimilarity function by switching the measurement onto a different basis (possibly a random one). As such, if $\mathcal{M}$ is the measurement basis, then the proximal point method will become:
\begin{equation}
\begin{gathered}
    \boldsymbol{\theta_{k+1}} =\text{prox}_{\mathcal{L}, \lambda, \mathcal{M}}(\boldsymbol{\theta_k}) \\
    = \arg \min_{\boldsymbol{\theta}}\left[ \mathcal{L}(\boldsymbol{\theta}) + \lambda q_{\mathcal{M}}(\boldsymbol{\theta}, \boldsymbol{\theta_k})\right].
\end{gathered}
\end{equation}
In that case, the update rule will be reformulated to:
\begin{equation}
    \boldsymbol{\theta_{k+1}} = \boldsymbol{\theta_k} - \eta \mathcal{F}_C^{\mathcal{M}}(\boldsymbol{\theta_k})^{+}\grad \mathcal{L}(\boldsymbol{\theta_k})
\label{eq:classical_fisher_update}
\end{equation}
where $\mathcal{F}_C^{\mathcal{M}}$ is the classical Fisher information matrix constructed by performing measurements on the $\mathcal{M}$ basis. When the measurement basis is chosen to be the computational basis state, then the update rule \eqref{eq:classical_fisher_update} is referred to in the classical optimization literature as the \emph{Natural Gradient Descent} (NGD). However, in a quantum setting, the CFIM is basis-dependent, and as such, this property offers a significant computational advantage compared to QNG.

From the previous discussion, we can immediately introduce Corollary \ref{corollary:CFIM}, which provides a condition so that the update \eqref{eq:classical_fisher_update} results in a descent direction. The condition under which a CFIM (constructed by measurements on a basis $\mathcal{M})$ preconditions a gradient descent direction and results in a decrease on the loss function is as follows.

\begin{corollary}
        Consider the update rule \eqref{eq:classical_fisher_update} and let $\mathcal{F}_C^{\mathcal{M}}(\boldsymbol{\theta_k})$ be the CFIM constructed by performing measurements in the $\mathcal{M}$ basis. Then, the updated direction will result in a descent direction ($\mathcal{L}(\boldsymbol{\theta_{k+1})} \leq \mathcal{L}(\boldsymbol{\theta_{k})}$) as long as:
    \begin{equation}
        \eta \mathcal{F}_C^{\mathcal{M}}(\boldsymbol{\theta_k})^{+}H \mathcal{F}_C^{\mathcal{M}}(\boldsymbol{\theta_k})^{+} \preccurlyeq 2\mathcal{F}_C^{\mathcal{M}}(\boldsymbol{\theta_k})^{+}
    \label{eq:condition_of_CFIM}
    \end{equation}
for some $\eta$>0, where $H$ is the Hessian of the loss function $\mathcal{L}$ at the point $\boldsymbol{\theta_k}$.
\label{corollary:CFIM}
\end{corollary}

Corollary \ref{corollary:CFIM} provides the condition so that a tuple $(\mathcal{F}_C^{\mathcal{M}}, \eta)$ will result in a decrease of the loss function. However, testing the condition in Eq. \ref{eq:condition_of_CFIM} cannot be feasibly implemented in a practical setting as it requires the calculation of the Hessian at each point $\boldsymbol{\theta_k}$. As such, we would have to rely on empirical choices for the choice of the hyperparameter $\eta$.

It is true that all directions that leave the quantum state invariant under translations of the parameters will also leave all probability distributions unaffected. The former directions correspond to the eigenvectors of the QFIM that belong in its \emph{null space} while the latter directions correspond to eigenvectors of different classical Fisher information matrices. We can thus show (see Proposition \ref{proposition:null_space}) that all zero eigenvalues of the QFIM correspond to zero eigenvalues of any CFIM (but not vice versa).

\begin{proposition}
\label{proposition:null_space}
The null space of the quantum Fisher information matrix $\mathcal{N}(\mathcal{F}_Q)$ is a subspace of the null space of the classical Fisher information matrix $\mathcal{N}(\mathcal{F}_C^{\mathcal{M}})$ over any measurement collection $\mathcal{M}$ at a fixed point $\boldsymbol{\theta}$
\begin{equation}
   \mathcal{N}(\mathcal{F}_Q) \subseteq \mathcal{N}(\mathcal{F}_C^{\mathcal{M}}).
\end{equation}
\end{proposition}
\begin{proof}
    The infidelity between two parameterized quantum states $\ket{\psi(\boldsymbol{\theta})}$ and $\ket{\psi(\boldsymbol{\theta}+\boldsymbol{\epsilon})}$ is given by (see Eq. \eqref{eq:quantum_fisher}):
    \begin{equation*}
    d_f(\ket{\psi(\boldsymbol{\theta})},\ket{\psi(\boldsymbol{\theta}+ \boldsymbol{\epsilon})}) = \frac{1}{2} \boldsymbol{\epsilon}^T \mathcal{F}_Q \boldsymbol{\epsilon}
    \end{equation*}
    where $\mathcal{F}_Q$ is the QFIM at point $\boldsymbol{\theta}$. Consider now the eigenvalue decomposition of $\mathcal{F}_Q$:
    \begin{equation}
        \mathcal{F}_Q = U D U^T
    \end{equation}
    where $U$ is a unitary matrix with its columns being the normalized eigenvectors of $\mathcal{F}_Q$ and $D$ the diagonal matrix with the eigenvalues of $\mathcal{F}_Q$ as its entries. The distance between the two states can then be written as:
    \begin{equation}
    d_f(\ket{\psi(\boldsymbol{\theta})},\ket{\psi(\boldsymbol{\theta}+ \boldsymbol{\epsilon})}) = \frac{1}{2} \boldsymbol{\epsilon}^T U D U^T \boldsymbol{\epsilon}.
    \end{equation}
    Now, if we assume that $\boldsymbol{\epsilon} = d\alpha \boldsymbol{v_i}$ where $\boldsymbol{v_i}$ is an eigenvector of $\mathcal{F}_Q$ with eigenvalue $\lambda_i$, then the distance can be written as:
    \begin{equation}
    d_f(\ket{\psi(\boldsymbol{\theta})},\ket{\psi(\boldsymbol{\theta}+ \boldsymbol{\epsilon})}) = \frac{\lambda_i}{2} (d\alpha)^2.
    \end{equation}
    Thus, all infinitesimal changes $\delta \boldsymbol{u} \in \mathcal{N}(\mathcal{F}_Q)$ that belong in the null space of $\mathcal{F}_Q$ will not result in a change of the underlying quantum state. A direct consequence is that since the quantum states remain invariant under these translations, the probability distributions over any measurements will also remain unchanged. As such, since the probability distributions do not change for small displacements $\delta \boldsymbol{u}$, then $\delta \boldsymbol{u} \in \mathcal{N}(\mathcal{F}_C^{\mathcal{M}})$ and thus:
    \begin{equation*}
    \mathcal{N}(\mathcal{F}_Q) \subseteq \mathcal{N}(\mathcal{F}_C^{\mathcal{M}}).
\end{equation*}
\end{proof}

At this point, we would like to stress that the converse is not true. CFIM (obtained by a general measurement) may have zero eigenvalues that are not zero eigenvalues of the QFIM. This also implies that different measurements lead to CFIM that have different null spaces, with only guarantee that all of them contain the null space of the QFIM. It follows that some measurements lead to CFIMs that carry more information about changes happening in the quantum state and are closer to the QFIM than other measurements.

Finally, we can see that as in our case, a direction $\propto \mathcal{F}_C^{\mathcal{M}}(\boldsymbol{\theta_k})^{+} \grad \mathcal{L}(\boldsymbol{\theta})$ points towards the steepest descent direction (see section \ref{subsubsection:steepest_descent}) in the Riemannian space whose metric is the CFIM constructed by performing measurements in the $\mathcal{M}$ basis. But as we discuss in Section \ref{sec:optimal_measuremnt}, all $[\mathcal{F}_C^{\mathcal{M}}]$ are information matrices that carry partial local information of the quantum state with respect to the choice of measurements. In other words, all CFIMs can be seen as providing \emph{local approximations of the geometry of the underlying state-space} with the quality of the approximation determined by its distance from the QFIM.

\section{Optimal Measurement}
\label{sec:optimal_measuremnt}

Consider again two parameterized quantum states, $\ket{\psi(\boldsymbol{\theta})}$, $\ket{\psi(\boldsymbol{\theta} + \boldsymbol{\epsilon})}$ that differ by a small shift vector $\boldsymbol{\epsilon} \in \mathbb{R}^m$. Suppose we perform a measurement on a basis $\mathcal{M}$. As we have already seen, the distance of the corresponding probability distributions $p_{\mathcal{M}}(\boldsymbol{\theta})$, $p_{\mathcal{M}}(\boldsymbol{\theta}+\boldsymbol{\epsilon})$ can be written as:
\begin{equation}
    \mathrm{KL}(p_{\mathcal{M}}(\boldsymbol{\theta})|| p_{\mathcal{M}}(\boldsymbol{\theta}+\boldsymbol{\epsilon})) = \frac{1}{2}\boldsymbol{\epsilon}^T \mathcal{F}_C^{\mathcal{M}}(\boldsymbol{\theta}) \boldsymbol{\epsilon}.
\label{eq:distance_fisher}
\end{equation}
Our goal is to choose a measurement basis $\mathcal{M}$ that will extract the maximum information from the state $\ket{\psi(\boldsymbol{\theta})}$. In the context of this paper, \emph{maximum information} refers to a measurement that will approximate as much as possible what happens locally in the space of quantum states.\\

\begin{definition} \textbf{(Optimal Measurement)}
We define the optimal measurement $\mathcal{M^*}$ as the measurement that when applied on the states $\ket{\psi(\boldsymbol{\theta})}$ and $\ket{\psi(\boldsymbol{\theta} + \boldsymbol{\epsilon})}$ it maximizes the distance between the probability distributions $p_{\mathcal{M}}(\boldsymbol{\theta})$ and $p_{\mathcal{M}}(\boldsymbol{\theta} + \boldsymbol{\epsilon})$.
\begin{equation}
\mathcal{M^*} = \arg\max_{\mathcal{M}}  \mathrm{KL}(p_{\mathcal{M}}(\boldsymbol{\theta})|| p_{\mathcal{M}}(\boldsymbol{\theta}+\boldsymbol{\epsilon})).
\end{equation}
\end{definition}

As we discussed in Sec \ref{sec:distance_of_prob_distr} we can describe the possible measurements $\mathcal{M}(\boldsymbol{\phi})$ by first applying a unitary $V(\boldsymbol{\phi})$ on the parameterized state $\ket{\psi(\boldsymbol{\theta})}$ and then performing projective measurements on each qubit. If $\boldsymbol{\phi^*}$ are the angles that maximize the distance between probability distributions, then the following inequality holds
\begin{equation}
\begin{aligned}
    \mathrm{KL}(p_{\mathcal{M}(\boldsymbol{\phi^*})}(\boldsymbol{\theta})|| p_{\mathcal{M}(\boldsymbol{\phi^*})}(\boldsymbol{\theta}+\boldsymbol{\epsilon})) \\
    \leq \mathrm{KL}(p_{\mathcal{M^*}}(\boldsymbol{\theta})||p_{\mathcal{M^*}}(\boldsymbol{\theta}+\boldsymbol{\epsilon}))
\end{aligned}
\end{equation}
where the equality is true whenever there exists $\boldsymbol{\phi^*}$  such that $\mathcal{M}(\boldsymbol{\phi}^*) = \mathcal{M}^*$. We can use Eq. \eqref{eq:distance_fisher} and show that any CFIM can be upper bounded by the QFIM (see \cite{meyer2021fisher} for details) as :
\begin{equation}
    \mathcal{F}_C^{\mathcal{M}}(\boldsymbol{\theta}) \preccurlyeq \mathcal{F}_Q(\boldsymbol{\theta}) \; \forall \boldsymbol{\theta}\in \mathbb{R}^m.
\label{eq:fish_upper_bound}
\end{equation}
A natural question to ask is \emph{``What is the appropriate measurement basis so that the resulting CFIM is optimal, in the sense that the CFIM approaches the QFIM with the least amount of error."}. The answer to this question is discussed below.

Consider a parameterized quantum circuit that generates parameterized quantum states $\ket{\psi(\boldsymbol{\theta})}$. Consider also the set of measurements that are generated by applying a unitary $V(\boldsymbol{\phi})$ on the state $\ket{\psi(\boldsymbol{\theta})}$ and then measuring in the computational basis.

Our starting point is Eq. \eqref{eq:fish_upper_bound}. For every angle configuration $\boldsymbol{\phi}$, the matrix $(\mathcal{F}_C^{\mathcal{M}(\boldsymbol{\phi})}(\boldsymbol{\theta}) -\mathcal{F}_Q(\boldsymbol{\theta}))$ is negative semidefinite for every $\boldsymbol{\theta}$. As such, by taking the trace:
\begin{equation}
\begin{gathered}
    \tr \left( \mathcal{F}_C^{\mathcal{M}(\boldsymbol{\phi})}(\boldsymbol{\theta}) - \mathcal{F}_Q(\boldsymbol{\theta})\right) \leq 0 \\
    \tr \left(\mathcal{F}_C^{\mathcal{M}(\boldsymbol{\phi})}(\boldsymbol{\theta}) \right) \leq \tr \left(\mathcal{F}_Q(\boldsymbol{\theta})\right)
    \end{gathered}
    \end{equation}
As a result, the trace of QFIM provides an upper bound on the trace of every CFIM. We can thus conclude that the solution of the optimization problem:
    \begin{equation}
            \max_{\boldsymbol{\phi}} \tr (\mathcal{F}_C^{\mathcal{M(\boldsymbol{\phi})}}(\boldsymbol{\theta}))
\label{eq:optimal_measurement_optimization}
    \end{equation}
will result in the CFIM corresponding to the optimal measurement (or else the optimal approximation of QFIM).

At this point, it is important to stress that except for the one-parameter case, there does not always exist a measurement basis $\mathcal{M}$ so that the CFIM is equal to QFIM \cite{braunstein1994statistical, liu2020quantum}. Specifically, in \cite{pezze2017optimal}, the authors provide conditions under which there exists measurement so that the QFIM is saturated. However, if no such measurement exists, then the solution of Eq. \eqref{eq:optimal_measurement_optimization} will result in a CFIM that approximates QFIM with the least amount of error:
\begin{equation}
    \min_{\mathcal{M}}\norm{\mathcal{F}_C^{\mathcal{M}}(\boldsymbol{\theta}) - \mathcal{F}_Q(\boldsymbol{\theta})}
\end{equation}
where $\norm{\cdot}$ is any matrix norm. Moreover, as we discussed in Section \ref{sec:local_optimization}, in a classical optimization scheme, achieving the optimal measurement may not have an immediate effect on the speed of convergence. Ideally, we would want to increase the number of directions in the probability distribution space so that the optimizer can have more directions to move.  As we will discuss in Section \ref{sec:random_natural_gradient}, this is accomplished when the measurement basis is chosen at random.

The immediate advantage of identifying the optimal measurement is that provided that the resulting CFIM coincides with QFIM,  the quantum natural gradient \cite{stokes2020quantum} can be implemented using $\mathcal{O}(m)$ quantum states/resources instead of $\mathcal{O}(m^2)$. This results in a quadratic advantage in quantum resource requirement compared to previous methods. However, when the measurement operators are parameterized by a unitary $V(\boldsymbol{\phi})$, the optimization problem becomes non-convex. On top of that, the optimal measurement is $\boldsymbol{\theta}$-dependent, and as such, after every optimization iteration, the optimal measurement must be re-evaluated. 

\textbf{Remark.} As we observed in our experiments, increasing the expressivity of the random parameterized measurements results in CFIMs that 1) tends towards $\frac{1}{2} \mathcal{F}_Q$, i.e. the error $\norm{\mathcal{F}_C^{\mathcal{M}} - \frac{1}{2}\mathcal{F}_Q}$ is reduced and 2) the variance of the error (of any random CFIM) goes to zero.

Consider, for example, the parameterized quantum circuit in the left-hand side of Figure \ref{fig:quantum_circuits} for a system of 8 qubits and 3 layers. As it is illustrated in Figure \ref{fig:histogram}, it is clear that a Pauli measurement cannot encapsulate the intrinsic changes happening in the quantum state. However, as we introduce random measurements, the random CFIM starts to approximate the QFIM, with the approximation becoming better when more expressive ansatz families are used. For the random measurements, we used the same type of circuit but with different Pauli rotations.

\begin{figure}
\begin{tikzpicture}
\node (img)  {\includegraphics[scale=0.6]{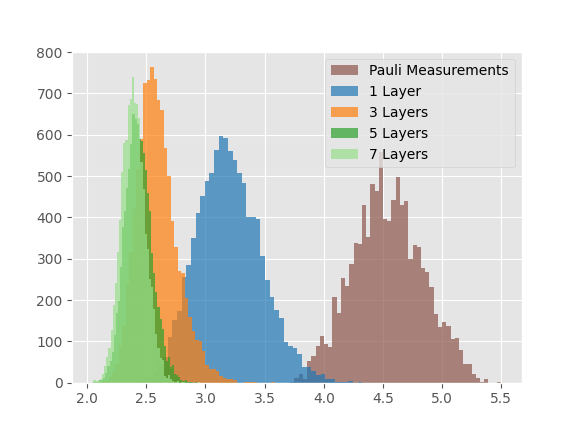}};
\node[below=of img, node distance=0cm, xshift=0.2cm, yshift=1.5cm] {\scriptsize $\norm{\mathcal{F}_C^{\mathcal{M}} - \frac{1}{2}\mathcal{F}_Q}$};
\node[left=of img, node distance=0cm, rotate=90, xshift=0.4cm, yshift=-1.5cm] {\scriptsize Samples};
\end{tikzpicture}
\caption{Distance of random CFIMs from QFIM. As the number of layers increases, sampling a random measurement tends to have a small distance from the QFIM and, as such, carries more information. A total number of 10000 CFIMs were calculated for each choice of measurement basis.}
\label{fig:histogram}
\end{figure}

\section{Stochastic-Coordinate Quantum Natural Gradient}
\label{sec:stochastic_coordinate}

In this section, we will provide our second main result which is based on a lower-rank approximation of the QFIM. During the past few years, researchers have proposed a number of ways to approximate the QFIM by reducing the quantum resources in order to make the QNG more applicable in real-world settings. As we previously mentioned, Stokes \emph{et al} \cite{stokes2020quantum} proposed that instead of calculating the full QFIM, one could calculate a block-diagonal approximation of the QFIM. However, such an approximation may not be valid when off-block diagonal terms are highly correlated.

On the other hand, in \cite{Gacon2021simultaneous}, the authors suggested that one could apply the 2-SPSA algorithm (which is used to calculate the Hessian of a loss function) to approximate the QFIM. The authors suggested that this strategy is efficient as it requires a constant number of quantum states independent of the number of parameters. However, this approximation requires a larger number of shots as the number of parameters increases (or a smaller step size during the optimization) in order to achieve the same accuracy.

In this section, we provide a new approximation that is inspired by coordinate descent algorithms \cite{wright2015coordinate}. In these algorithms, the user determines a coordinate \cite{nesterov2012efficiency}, or a block of coordinates \cite{tseng2001convergence} that will update on each iteration and keeps all other directions fixed. At this point, we need to take a step back and discuss a redundancy measure that was introduced in \cite{haug2021capacity}. Consider a PQC $C$ with $m$ parameters and let its parameter dimension $D_C$. As introduced in \cite{haug2021capacity}, the parameter dimension $D_C$ quantifies the number of independent parameters that the PQC can express in the space of quantum states. Let also $G_C(\boldsymbol{\theta})$ be the rank of QFIM at point $\boldsymbol{\theta}$. The authors numerically verified that for PQCs in which their parameterized gates follow a $(0, 2\pi)$ gate periodicity:
\begin{equation}
    G_C(\boldsymbol{\theta}) \approx D_C
\end{equation}
for randomly chosen $\boldsymbol{\theta}$. As such, for a given PQC, by measuring the rank of QFIM at random points, we can calculate the redundancy of the parameters
\begin{equation}
    R = \frac{m - G_C(\boldsymbol{\theta})}{m}.
\end{equation}
In Figure \ref{fig:ranks}, we illustrate the ranks of both the QFIM and CFIMs compared to the total number of parameters. We can visualize that the redundancy measure is large and that only a fraction of the total parameters contribute to changing the quantum state in an independent way.

Consider the QFIM at a given configuration $\boldsymbol{\theta}$, with rank $m-k$ where $k$ is the dimension of its kernel and $m$ is the total number of parameters. In the ideal case where the eigenvectors and the eigenvalues are known, identifying the parameters that can change the quantum state can be performed using the following procedure.

Let the kernel of $\mathcal{F}_Q$ spanned by the eigenvectors $\{\boldsymbol{v_1}, \ldots, \boldsymbol{v_k}\}$ with $k<m$. If we project the parameters onto that subspace, we can identify which parameters matter most. For example, if a given parameter $\theta_i$ has the largest projection onto that subspace, then we can deduce that varying this parameter will result in minor (if not negligible) changes in the underlying quantum state. Moreover, if the projection onto the space orthogonal to the kernel is zero, then varying this parameter will have no effect on the quantum state. Mathematically, as also noted by \cite{haug2021capacity}, one would have to calculate the quantity:
\begin{equation}
    g_i = \sum_{j=1}^{k} |v_j^i|^2
\end{equation}
for every $i\in [m]$ and the largest $g_i$ would correspond to parameters whose variation will not lead to any change of the state. However, the problem is that such a procedure is inefficient in practice, as different parameters may matter most in different parameter settings, and the calculation of the full QFIM is needed. Instead, we propose the following solution to this problem, which is computationally cheaper but may not always find all the parameters that result in an independent change.

\begin{figure}
\begin{tikzpicture}
\node (img)  {\includegraphics[scale=0.5]{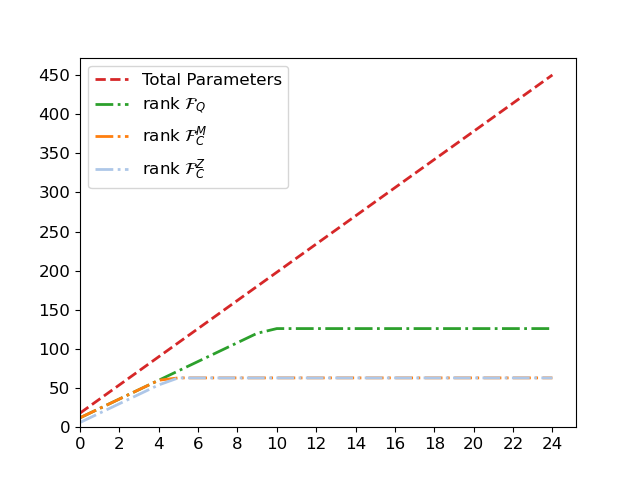}};
\node[below=of img, node distance=0cm, yshift=1.5cm] {\scriptsize Layers};
\end{tikzpicture}
\caption{Ranks of QFIM $\mathcal{F}_Q$, CFIM with $Z$-basis measurements $\mathcal{F}_C^Z$, and CFIM with measurements on a random basis $\mathcal{F}_C^\mathcal{M}$ compared to the total number of parameters of the parameterized quantum circuit on the left of Figure \ref{fig:quantum_circuits}.}
\label{fig:ranks}
\end{figure}

Now, consider a subset $L\subseteq [m]$ (of cardinality $|L| = l\leq m$) of the total number of parameters and let $\mathcal{F}_{RQ}$ be the \emph{reduced} QFIM with elements defined as:
\begin{equation}
    [\mathcal{F}_{RQ}]_{ij} = \begin{cases}
        [\mathcal{F}_Q]_{ij} \; & \text{if $i,j \in L$}\\
        0 \; &\text{otherwise}
    \end{cases}
\label{eq:reduced_QFIM}
\end{equation}
where $[\mathcal{F}_Q]_{ij}$ are the elements of the QFIM defined in Eq. \eqref{eq:qfim_elements}. We can immediately see that the first advantage of this approximation is that the cost of calculating the reduced QFIM immediately drops down to $\mathcal{O}_Q(l^2)$ quantum state preparations (but the same classical memory resources). The second advantage is that increasing the size of the subset $L$, i.e. considering more parameters, improves the accuracy of the approximation, but increases the cost, with the method essentially becoming QNG when $L = [m]$. The physical intuition behind this approximation is that the reduced QFIM carries information about how the quantum state changes (with respect to a distance measure) if we vary only a portion of the total parameters. However, the question that naturally arises is how one can pick any such coordinate subset.

A straightforward way to choose the coordinate subset is to sample uniformly a subset (of user-specified cardinality) of the total number of parameters. The probability that this subset includes all independent parameters is given by Lemma \ref{lemma:probability}.

\begin{lemma}
    Consider a parameterized quantum circuit $C$ composed of $m$ parameters. Consider also the unknown subset $L \subseteq [m]$ (of cardinality $ |L| = l$) of parameters whose variation results in an independent change of the underlying quantum state at point $\boldsymbol{\theta}$. Let also $S_k \subseteq [m]$ be the set of (uniformly) randomly sampled parameters of cardinality $|S_k| = l+k \leq m$. Then, the probability of $L\subseteq S_k$ is
    \begin{equation}
        \mathbb{P}[L\subseteq S_k] = \frac{(l+k)!}{k!l!} \frac{l!(m-l)!}{m!}.
    \end{equation}
\label{lemma:probability}
\end{lemma}

\begin{proof} The proof can be found in Appendix \ref{appendix:proof_lemma}
\end{proof}

One can use this approximation and construct an approximation to the QNG. At each iteration, the user samples (uniformly at random) a different subset $L_i \subset[m]$ of the total coordinates and calculates the reduced QFIM given by Eq. \eqref{eq:reduced_QFIM}. Then, the parameters are updated according to the rule:
\begin{equation}
    \boldsymbol{\theta}_{k+1} = \boldsymbol{\theta}_k - \eta [\mathcal{F}_{RQ}(\boldsymbol{\theta}_k)]^{+} \grad{\mathcal{L}}(\boldsymbol{\theta}_k).
\label{update_rule_rqng}
\end{equation}

\begin{algorithm}
\caption{Stochastic-Coordinate Quantum Natural Gradient}
\label{alg:stochastic_coordinate}
\SetKwInOut{Input}{Input}
\Input{Problem Hamiltonian $\mathcal{H}$\;
Ansatz family $\ket{\psi(\boldsymbol{\theta})}=U(\boldsymbol{\theta})\ket{0}$\;
Total iterations $K$\;
Loss function $\mathcal{L}(\boldsymbol{\theta})$\;
Initial parameters $\boldsymbol{\theta} = \boldsymbol{\theta_0}$\;
Learning rate $\eta$\;}
\For{$k=1,2,\ldots, K$}{
Shuffle random subset of coordinates $L_k \subset [m]$\;
Calculate derivatives $\frac{\partial \mathcal{L}(\boldsymbol{\theta})}{\partial \theta_i} \; \forall \; i \in L_k$\;
Calculate the reduced QFIM $\mathcal{F}_{RQ}(\boldsymbol{\theta})$ \;
Update $\boldsymbol{\theta}$ as $\boldsymbol{\theta} = \boldsymbol{\theta} - \eta [\mathcal{F}_{RQ}(\boldsymbol{\theta})]^{+} \grad{\mathcal{L}}(\boldsymbol{\theta})$\;
}
\Return $\boldsymbol{\theta}$
\end{algorithm}

\section{Method Evaluation}
\label{sec:method_evaluation}
In the first part of this section, we discuss the mathematical problems used for our experiments. In the second part, we discuss the appropriate figures of merit used to benchmark our proposed algorithms (see Random Natural Gradient in Sec. \ref{sec:random_natural_gradient} and Stochastic-Coordinate Quantum Natural Gradient in Sec. \ref{sec:stochastic_coordinate}). The choice of quantum circuits used in our experiments is given in Appendix \ref{appendix:parameterized_circuits}.

\subsection{Mathematical Problems}

We consider three problems: two classical (MaxCut and Number Partitioning) and one quantum (Heisenberg model).

\subsubsection{MaxCut}

The first mathematical problem that we tackled is a classical combinatorial optimization problem called \emph{MaxCut.} In this task, the user is presented with a non-directed $n$-vertex graph $G(V, E)$ where $V$ is the set of vertices, and $E$ is the set of edges. Each edge connecting vertex $i$ with vertex $j$ is weighted by a non-zero weight $w_{ij}$.

A \emph{cut} is defined as a partition of the set of vertices into two disjoint subsets $V_1, V_2$ (with $V_1\cup V_2 = V$ and $V_1 \cap V_2 = \emptyset$). Equivalently, we label $0$ each vertex that belongs in the set $V_1$ and $1$ otherwise. The target of the MaxCut problem is to maximize the cost function:
 \begin{equation}
     C(\boldsymbol{x}) = \sum_{i,j = 1}^n w_{ij}x_i\left(1-x_j\right).
 \end{equation}
 where $x_i\in \{0,1\}$ are binary variables. Intuitively, this corresponds to finding a partition of the vertices into two disjoint sets that ``cuts'' the maximum number of edges. By transforming the binary variables $x_i$ to spin variables $z_i$ according to $x_i = \frac{1-z_i}{2}$, the cost function is mapped into a quantum spin-configuration problem,
 \begin{equation}
     C(\boldsymbol{z}) = \sum_{\left<i,j\right>\in E} \frac{w_{ij}}{2}\left(1-z_iz_j\right).
 \end{equation}
Maximizing the cost function above corresponds to finding the ground state of the Hamiltonian
\begin{equation}
	H_{\textrm{MC}} = -\sum_{\left<i,j\right>\in E}\frac{w_{ij}}{2}\left( 1- \sigma_{i}^z \sigma_j^z \right).
 \label{eq:maxcut_hamiltonian}
\end{equation}

MaxCut has been extensively used in the VQA literature \cite{farhi2014quantumapproximateoptimizationalgorithm, wang2018quantum,  crooks2018performance, bravyi2020obstacles} as a benchmarking problem to test the performance of certain near-term algorithms and techniques, and for that reason, we will also examine the performance of our algorithms in this setting.

\subsubsection{Number Partitioning}

The second problem that we chose to tackle is the \emph{Number Partitioning} problem, a classical problem that has been previously used as a benchmark in the VQE literature \cite{nannicini2019performance, barkoutsos2020improving, kolotouros2022evolving}. The Number-Partitioning problem can be mapped to the \emph{Sherrington-Kirkpatrick model}, which has been previously analysed in \cite{farhi2022quantum}. The problem is stated as follows.

We are given a set of $N$ integers $\{n_1, n_2, \ldots, n_N\}$ and we are asked to decide whether there exists a partition of the set into two disjoint subsets $S, \bar{S}$ so that the sums of the elements on each subset are equal. The above problem can be cast as an optimization problem with a cost function:
\begin{equation}
    C(\boldsymbol{x}) = \left(\sum_{i=1}^N (2x_i - 1)n_i\right)^2.
\end{equation}
The binary string $\boldsymbol{x}=x_1x_2\dots x_n$ corresponds to one configuration where a number $n_i$ is placed in the $S$ set ($x_i=0$) or in the $\bar{S}$ set ($x_i = 1)$. The Number Partitioning problem can then easily be formulated as a quantum spin problem (see \cite{lucas2014ising}) interacting with Hamiltonian:
\begin{equation}
    H_{\textrm{NP}} = \sum_{i\neq j}(n_i n_j)\sigma_i \sigma_j + \sum_{i=1}^N n_i^2.
\end{equation}

\subsubsection{Heisenberg Model}

The second mathematical problem that we tackle is to find the ground state of a quantum Hamiltonian. Specifically, our goal is to prepare the ground state of a \emph{Heisenberg Model}, which corresponds to a Hamiltonian having extra off-diagonal terms. The following Hamiltonian describes the 1D XXX Heisenberg model:
\begin{equation}
    H_{XXX} = J\sum_{i=1}^{n-1}\left(\sigma_i^x \sigma_{i+1}^x + \sigma^y_i \sigma^y_{i+1} + \sigma^z_i \sigma^z_{i+1}\right) + h \sum_{i=1}^n \sigma_x^i
\end{equation}

The Heisenberg model has also been extensively used in the VQA literature for benchmarking purposes. In \cite{gacon2023variational} the authors simulated \emph{quantum imaginary time evolution} to prepare ground states of the Heisenberg model while in \cite{wierichs2020avoiding} the authors investigated how QNG avoids local minima in a variation of the Heisenberg model, called Transverse-Field Ising Chain model. 

\subsection{Evaluation Metrics}

In order to fairly evaluate our proposed algorithms, we choose to benchmark based on two different metrics. In a hybrid quantum-classical setting, one would have to carefully evaluate both the classical and the quantum resources needed to execute an optimization algorithm.

As we compare classical optimization algorithms, a fair choice of metric is to count the number of optimization iterations (or else how many times we have to update the parameters of the quantum circuit) until convergence.  

The second choice of metric is how many (different) quantum states we have to prepare until we converge. For example, as we already discussed, although the QNG may converge faster in terms of optimization iterations, it actually requires $\mathcal{O}_Q(m^2)$ quantum states at each iteration. So, a careful analysis may prohibit any practical advantage, as the resources may be larger than performing naive gradient descent.

On top of that, it is beneficial to benchmark our method to QNG and GD on a large number of instances in order to evaluate the overall performance. For this reason, we choose two additional metrics. The first (which quantifies the quality of the output solution) is the average probability of sampling the optimal solution (i.e. the overlap with the ground state) of the output state (when the optimization is terminated). The second (which quantifies the speed of methods) is the average relative error (over all different instances) per iteration. It quantifies how much, on average, at a given iteration, the current state differs from the optimal state.

\section{Results}
\label{sec:results}

In this section, we will illustrate how the two proposed methods perform in different optimization settings. In the first part, we will investigate the performance of the \emph{Random Natural Gradient} (see Sec \ref{sec:random_natural_gradient}) while in the second part, the \emph{Stochastic-Coordinate Quantum Natural Gradient} (see Sec. \ref{sec:stochastic_coordinate}).

\subsection{Technical Details}

In this manuscript, all simulations were performed using Qiskit's exact \emph{Statevector simulator}, which allows noiseless executions of the quantum circuits. For all classical optimization algorithms used, we calculated the gradients using the parameter shift rules \cite{Schuld_2019} and used the same learning rate $\eta\in [0.001, 0.1]$ for all algorithms. For the classical combinatorial optimization problems, we used the quantum circuit on the right of Figure \ref{fig:quantum_circuits} with nearest-neighbours interactions and 4 layers while for the Transverse-Field Ising problem, the quantum circuit on the left of Figure \ref{fig:quantum_circuits} with an all-to-all connectivity and 3 layers. Moreover, for the random measurements, we used the same type of circuits with different Pauli rotations and with smaller depths than the circuits that generated $\ket{\psi(\boldsymbol{\theta})}$. Finally, in order to calculate the Moore-Penrose inverses for the update steps, we set a threshold of $10^{-4}$ as a cutoff for the singular values in order to prohibit very large steps.

\subsection{Random Natural Gradient}

\subsubsection{MaxCut}

You can illustrate the overall performance of RNG compared to GD and QNG in Table \ref{maxcut_table} and Figure \ref{fig:relative_error_comparison}. The three methods were compared on random weighted 3-regular graphs. This class of graphs have only two optimal solutions, where each one can be acquired from the other by flipping all qubits (due to the $\mathbb{Z}_2$ symmetry of the MaxCut problem).

In Table \ref{maxcut_table}, we can illustrate the probability of sampling the optimal solution when the three different optimization methods were used with the same initial angles and the same step size ($\eta=0.05$). We can see that the Random Natural Gradient and Quantum Natural Gradient perform almost similarly. The Gradient Descent method, however, always performs worse than the former information-theoretic methods. Overall, we expect RNG and QNG to perform similarly and, in the limit of infinitely many iterations, to converge to the same point (provided that the approximation of RNG to QNG is sufficient).

On the other hand, in Figure 3, we plot the average relative error of the output solution for the three methods (over the 30 different instances on 12 qubits). As can be clearly seen, the QNG and RNG return solutions that, on average, are always closer to the optimal solution. This highlights the fact (in agreement with \cite{wierichs2020avoiding}) that information-theoretic methods perform significantly better than methods that do not consider the information of the underlying state space. For additional experiments that illustrate the actual resources needed for RNG compared to QNG, we point the reader to Appendix \ref{sec:additional_experiments}.

\subsubsection{Number Partitioning}

As in MaxCut, we illustrate the overall performance of RNG compared to GD and QNG in Table \ref{number_partitioning_table}. For our experiments, we chose to sample integers from the $[0, 25]$ set and set the step-size of $\eta=0.001$ (a smaller step-size is needed to guarantee convergence as the cost values are much larger than MaxCut). Similar to MaxCut, Number Partitioning also has a $\mathbb{Z}_2$ symmetry. As it can be clearly visualized, the superiority of information methods is also present in the Number Partitioning problem. Clearly, the two information-theoretic optimization methods are able to return high-quality outputs, achieving significantly larger overlap with optimal solutions (on an average of 60 instances).

\subsubsection{Heisenberg Model}

You can visualize the performance of the Random Natural Gradient on a Heisenberg-model instance of 10 qubits (with couplings $J=h=1$) in Figure \ref{fig:random_natural_gradient}. For this instance, we used the hardware-efficient ansatz seen on the left side of Figure \ref{fig:quantum_circuits}, with a nearest-neighbour connection.

On the left side of Figure \ref{fig:random_natural_gradient}, we illustrate the number of optimization iterations (or else how many times we update the parameters) until convergence (we stop the optimization after 500 iterations). We see that the RNG is able to reach the region of the local minimum much faster (in terms of optimization iterations) compared to the QNG. On the other hand, the GD optimizer, as it doesn't carry any information about the underlying Riemannian space, gets stuck at a local minimum and performs significantly worse than QNG and RNG.

However, the biggest advantage is illustrated on the right-hand side of Figure \ref{fig:random_natural_gradient}. There, we can visualize the actual (quantum) resources needed until convergence. It is clear that the RNG offers a significant advantage in the number of quantum calls, reducing the overall overhead in current quantum devices (requiring almost ten times less quantum state preparations than the QNG until convergence). 

\begin{table}
\begin{center}
\begin{adjustbox}{width=\columnwidth,center}
\begin{tabular}{||c||c|c||}
    \hline
     \textbf{MaxCut} &\multicolumn{2}{c||}{Optimal Solution Overlap ($\%$)}\\
     \hline
     & 12 Qubits & 14 Qubits\\
     \hline
     Gradient Descent & 37.3 & 24.99 \\
     \hline
     Random Natural Gradient & 41.89 & 32.02\\
     \hline
     Quantum Natural Gradient & 42.4 & 29.99\\
     \hline
     
\end{tabular}
\end{adjustbox}
\caption{Probability of sampling the optimal solution for the MaxCut problem for 60 different instances (30 for 12 qubits and 30 for 14 qubits). All instances correspond to random 3-regular weighted graphs. Both QNG and RNG outperform GD.}
\label{maxcut_table}
\end{center}
\end{table}

\begin{table}
\begin{center}
\begin{adjustbox}{width=\columnwidth,center}
\begin{tabular}{||c||c|c||}
    \hline
     \textbf{Number Partitioning} &\multicolumn{2}{c||}{Optimal Solution Overlap ($\%$)}\\
     \hline
     & 12 Qubits & 14 Qubits\\
     \hline
     Gradient Descent & 20.83 & 17.62 \\
     \hline
     Random Natural Gradient & 27.54 & 25.54\\
     \hline
     Quantum Natural Gradient & 28.17 & 26.31\\
     \hline
\end{tabular}
\end{adjustbox}
\caption{Probability of sampling the optimal solution for the Number Partitioning problem for 60 different instances (30 for 12 qubits and 30 for 14 qubits). All instances correspond to a set of integers drawn from the $[1,25]$ interval. Both information-theoretic methods outperform GD.}
\label{number_partitioning_table}
\end{center}
\end{table}

\begin{figure}[!ht]
\begin{tikzpicture}
\node (img1)  {\includegraphics[scale=0.54]{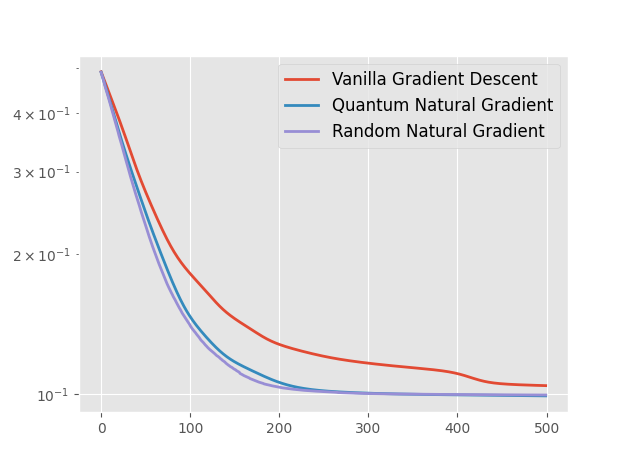}};
\node[below=of img1, node distance=0cm, yshift=1.3cm] (xlabel1) {\scriptsize Optimization Iterations};
\node[left=of img1, xshift=0.5cm, node distance=0cm, rotate=90, anchor=center,yshift=-0.5cm] {\scriptsize $(\mathcal{L}(\boldsymbol{\theta}) - E_{\text{opt}})/(E_{\text{opt}})$};
\end{tikzpicture}
\caption{Relative error for Gradient Descent, Random Natural Gradient and Quantum Natural Gradient for 30 12-qubit random weighted 3-regular graphs (in logarithmic scale).}
\label{fig:relative_error_comparison}
\end{figure}

\begin{figure*}
\begin{tikzpicture}
\node (img1)  {\includegraphics[scale=0.54]{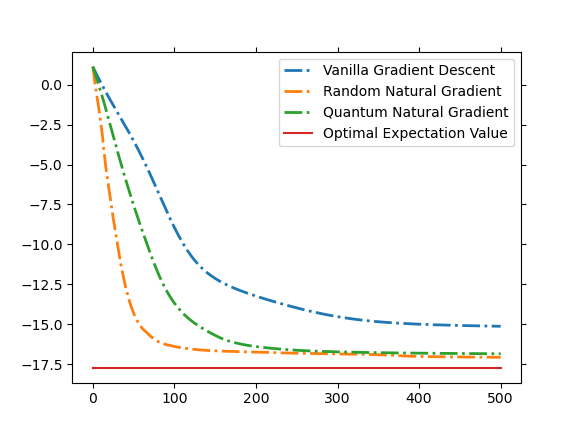}};
\node[below=of img1, node distance=0cm, yshift=1.3cm] (xlabel1) {\scriptsize Optimization Iterations};
\node[left=of img1, xshift=0.5cm, node distance=0cm, rotate=90, anchor=center,yshift=-0.7cm] {\scriptsize $\mathcal{L}(\boldsymbol{\theta})$};
\node[right=of img1, xshift=-1.6cm] (img2)  {\includegraphics[scale=0.54]{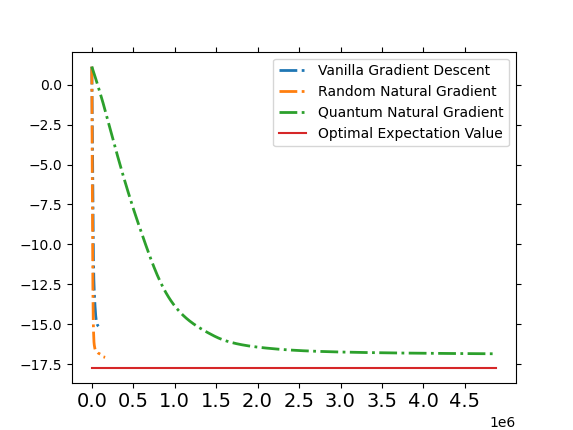}};
\node[below=of img2, node distance=0cm, yshift=1.1cm] (xlabel2) {\scriptsize Quantum State Preparations};
\node[left=of img2, xshift=0.5cm, node distance=0cm, rotate=90, anchor=center, yshift=-0.7cm] {\scriptsize $\mathcal{L}(\boldsymbol{\theta})$};
\end{tikzpicture}
\caption{Performance of the Random Natural Gradient optimizer on a Heisenberg model of 10 qubits compared to the Quantum Natural Gradient and Gradient Descent on both the optimization iterations (left figure) and on quantum resources (right figure). The RNG and GD methods require less quantum resources to converge (compared to QNG), but, as seen in both figures, the GD method converges to a bad quality minimum.}
\label{fig:random_natural_gradient}
\end{figure*}
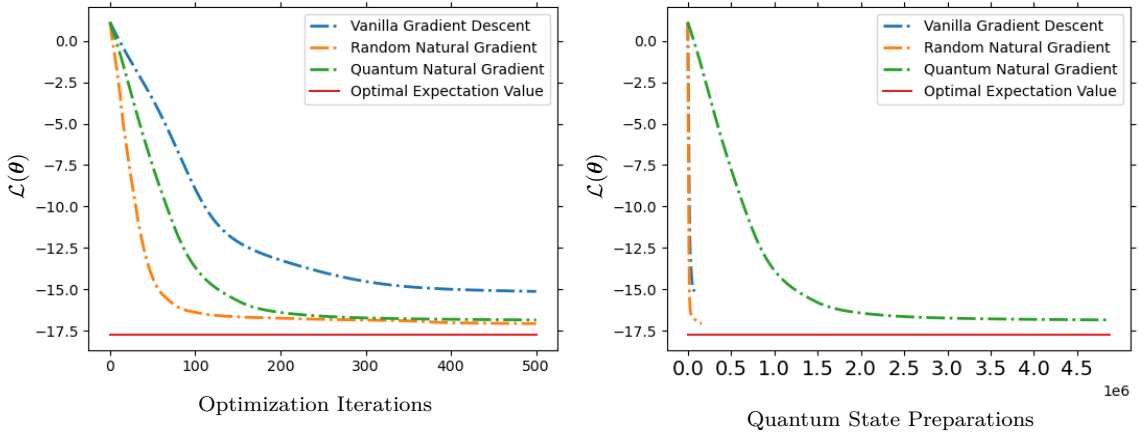

\subsection{Stochastic-Coordinate Quantum Natural Gradient}

In figure \ref{fig:scqng}, we can visualize the performance of the Stochastic-Coordinate Quantum Natural Gradient on a Heisenberg-model instance of 10 qubits (again with couplings $J=h=1$ but with different random initial angles).

An important hyperparameter in the SC-QNG optimizer is the cardinality of the random subset $S_k$ that we uniformly sample at each iteration. In this work, we make the naive choice that the user samples $m/2$ parameters at random in each iteration. As we can see in Figure \ref{fig:scqng}, both QNG and SC-QNG require the same number of optimization iterations to converge. However, the latter results in a significant reduction in the actual quantum resources needed, requiring only a fraction of the quantum states that we need to prepare for QNG. 

You can find additional experiments, with a comparison between the two proposed methods on MaxCut instances in Appendix \ref{sec:additional_experiments}.

\begin{figure*}[t]
\begin{tikzpicture}
\node (img1)  {\includegraphics[scale=0.55]{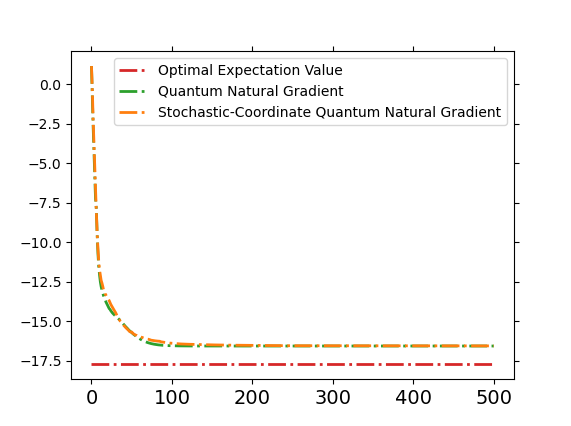}};
\node[below=of img1, node distance=0cm, yshift=1.2cm] (xlabel1) {\scriptsize Optimization Iterations};
\node[left=of img1, xshift=0.6cm, node distance=0cm, rotate=90, anchor=center,yshift=-0.7cm] {\scriptsize $\mathcal{L}(\boldsymbol{\theta})$};
\node[right=of img1, xshift=-1.6cm] (img2)  {\includegraphics[scale=0.55]{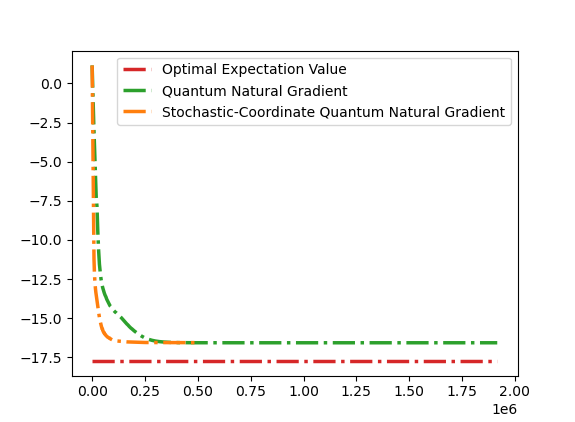}};
\node[below=of img2, node distance=0cm, yshift=1.2cm] (xlabel2) {\scriptsize Quantum State Preparations};
\node[left=of img2, xshift=0.6cm, node distance=0cm, rotate=90, anchor=center, yshift=-0.7cm] {\scriptsize $\mathcal{L}(\boldsymbol{\theta})$};
\end{tikzpicture}
\caption{Comparison of SC-QNG (with sampling half of the total parameters at each iteration) with QNG both in terms of optimization iterations (left) and quantum calls (right).}
\label{fig:scqng}
\end{figure*}
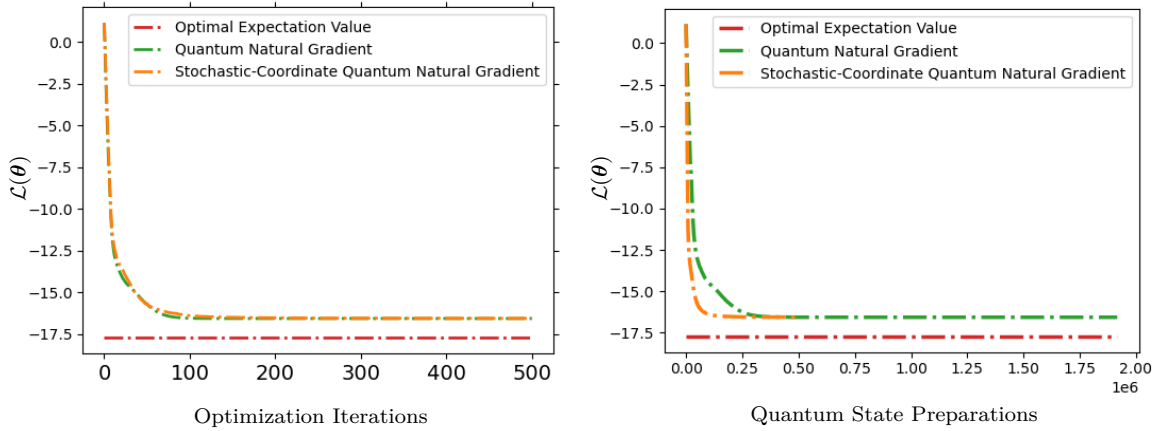

\section{Discussion}
\label{sec:discussion}

In this paper, we introduced two novel classical optimization algorithms that move Variational Quantum Algorithms (VQAs) a step closer to the practical regime. The way that we achieved this is by reducing the overall computational overhead in the quantum processor while at the same time exploiting quantum properties. 

The starting point of the first optimization algorithm is the Quantum Natural Gradient (QNG), introduced by \emph{Stokes et al.} \cite{stokes2020quantum}. In this algorithm, the parameters of a quantum circuit are updated in such a way that they still consider the changes happening locally in the quantum state. However, such an optimization procedure is impractical as the number of quantum calls grows as $\mathcal{O}(m^2)$, where $m$ is the number of parameters of the parameterized quantum circuit. To address this problem, we proposed \emph{Random Natural Gradient}, a classical optimization algorithm that requires quadratically less quantum calls than the QNG. Instead of performing the tedious calculation of the quantum Fisher information matrix (QFIM) at every iterate, we apply a random measurement and calculate the random classical Fisher information matrix (CFIM). The latter matrix provides an approximation to the changes happening in the quantum states, and as we find numerically, the approximation becomes better as we increase the expressiveness of the random unitary we apply for the random measurement.

The second optimization algorithm was inspired both from QNG and classical coordinate descent methods. In these methods, the user selects (possibly at random) a subset of the total parameters to optimize at each iteration reducing the overall computational overhead. As pointed out by \cite{haug2021capacity}, in Variational Quantum Circuits, not all parameters can lead to independent changes in the underlying quantum state. For that reason, we argue that one can construct a \emph{reduced} QFIM by examining how a subset of the total parameters changes a quantum state (with respect to a distance measure). This new matrix requires a fraction of the total overhead required by QFIM. We then argue that a user can randomly select and optimize a (different) subset of the total parameters at each iteration and we call this optimization \emph{Stochastic-Coordinate Quantum Natural Gradient} (SC-QNG).

Our results open up new directions of research in the optimization part of VQAs. As far as the RNG is concerned, one may wonder whether we can use fixed Pauli measurements instead of random measurements as defined earlier, reducing the extra depth in the parameterized quantum circuit. However, as we illustrate in Appendix \ref{appendix:cfim_pauli_measurements}, such an approach does not work. The reason is that both the error $\norm{\mathcal{F}_C^{\text{Pauli}} - \mathcal{F}_Q}$ and the number of singularities in the parameterized space is large for Pauli measurements (see also Figure \ref{fig:histogram}). As a result, the classical optimization algorithm will make very large steps, being unable to converge. The appropriate strategy is to use random measurements so that their corresponding Fisher information matrix is still close to the quantum Fisher information matrix. This opens up a new research direction, as the way of choosing the appropriate basis may also be viewed as a quantum machine learning problem. 

Another interesting research direction is to use additional figures of merit to quantify the practicality of RNG. In practical scenarios, one would have to use several circuit repetitions (shots) to accurately estimate the random classical Fisher information matrix and the gradient of the loss function. As explained in \cite{van2021measurement}, in certain scenarios, the calculation of the QFIM introduces negligible overhead compared to the gradient of the loss. As such, it would be fruitful to understand how much additional overhead is required to accurately achieve a good approximation of a random classical Fisher information matrix.

Additionally, a very interesting direction is to understand how information-theoretic methods perform in the presence of noise. It is clear that as noise starts to kick in, the information matrices lose the information they carry as errors start to accumulate in the matrix elements. For this reason, it would be fruitful to investigate how noise-tolerant these methods are and how well they perform against other (non-information-theoretic) optimization methods.

Finally, we would like to address the fact that as the number of qubits or the number of parameters increases, the error in the approximation of the QFIM by a random CFIM may also increase. A possible solution to circumvent this approach may be to use adaptive step sizes so that the condition that the optimizer moves into a descent direction still holds. This is also true for other classical optimization algorithms, such as the SPSA algorithm \cite{Gacon2021simultaneous} in which as the number of parameters increases, smaller step-sizes are necessary to guarantee convergence \cite{abbas2023quantum}.

On the other hand, for the SC-QNG, an immediate next step is to examine more sophisticated (but computationally inexpensive) ways to choose the randomly selected subset. For example, there may be quantities that can be calculated fast and can give information about which parameters matter most in a parameterized quantum circuit. Finally, the last thing that is essential in the VQA literature is to have a theoretical motivation about which optimization algorithms are suited for certain problems, as there may be quantum problems where the use of QNG (or approximations of it) is necessary to converge to a non-local minimum.

\section{Code Availability}
You can find a Python (Qiskit) implementation of both \emph{Random Natural Gradient} and \emph{Stochastic-Coordinate Quantum Natural Gradient} in \href{https://github.com/ioankolot/Random_Natural_Gradient}{Github}.

\section{Acknowledgements}
PW acknowledges support by EPSRC grants EP/T001062/1, EP/X026167/1 and EP/T026715/1, STFC grant ST/W006537/1 and Edinburgh-Rice Strategic Collaboration Awards. 

\bibliographystyle{quantum}
\bibliography{References}

\onecolumn\newpage
\appendix

\section{Derivation of Classical Fisher Information Matrix}
\label{sec:classical_fisher_derivation}

\noindent We start with the KL-divergence for two probability distributions $p(\boldsymbol{\theta})$, $p(\boldsymbol{\theta} + \boldsymbol{\epsilon})$:
\begin{equation}
\begin{gathered}
        \mathrm{KL}(p(\boldsymbol{\theta})||p(\boldsymbol{\theta} + \boldsymbol{\epsilon})) = \sum_{l} p_l(\boldsymbol{\theta}) \log \frac{p_l(\boldsymbol{\theta})}{p_l(\boldsymbol{\theta} + \boldsymbol{\epsilon})} =\\
        \sum_{l} p_l(\boldsymbol{\theta}) \log p_l(\boldsymbol{\theta}) - \sum_{l} p_l(\boldsymbol{\theta}) \log p_l(\boldsymbol{\theta}+\boldsymbol{\epsilon}). 
\end{gathered}
\end{equation}
The elements of the CFIM are defined as the second-order derivatives of the KL-divergence. Specifically, an element $[\mathcal{F}_C]_{ij}$ is given by:
\begin{equation}
\begin{gathered}
    [\mathcal{F}_C]_{ij} = -\frac{\partial^2}{\partial \epsilon_i \partial \epsilon_j} \sum_{l} p_l(\boldsymbol{\theta}) \log p_l(\boldsymbol{\theta}+\boldsymbol{\epsilon})\Bigg{|}_{\boldsymbol{\epsilon}=0} = - \sum_{l} p_l(\boldsymbol{\theta}) \frac{\partial^2}{\partial \epsilon_i \partial \epsilon_j} \log p_l(\boldsymbol{\theta}+\boldsymbol{\epsilon})\Bigg{|}_{\boldsymbol{\epsilon}=0} =\\
    -\mathbb{E}\Bigg\{\frac{\partial^2}{\partial \epsilon_i \partial \epsilon_j} \log p_l(\boldsymbol{\theta}+\boldsymbol{\epsilon})\Bigg{|}_{\boldsymbol{\epsilon}=0}\Bigg\} =  -\mathbb{E}\Bigg\{\frac{\partial^2}{\partial \theta_i \partial \theta_j} \log p_l(\boldsymbol{\theta})\Bigg\}
\end{gathered}
\label{eq:fisher_element_derive}
\end{equation}
since:
\begin{equation}
    \begin{gathered}
        -\frac{\partial^2}{\partial \theta_i \partial \theta_j} \log p_l(\boldsymbol{\theta}) = - \frac{1}{p_l(\boldsymbol{\theta})}\frac{\partial^2 p_l(\boldsymbol{\theta})}{\partial \theta_i \partial \theta_j} + \frac{1}{p_l^2(\boldsymbol{\theta})}\frac{\partial p_l(\boldsymbol{\theta})}{\partial \theta_i}\frac{\partial p_l(\boldsymbol{\theta})}{\partial \theta_j}.
    \end{gathered}
    \label{eq:derivative_fish}
\end{equation}
Substituting Eq. \eqref{eq:derivative_fish} in Eq. \eqref{eq:fisher_element_derive} we get:
\begin{equation}
\begin{gathered}
    [\mathcal{F}_C]_{ij} = \mathbb{E}\Bigg\{ -\frac{1}{p_l(\boldsymbol{\theta})}\frac{\partial^2 p_l(\boldsymbol{\theta})}{\partial \theta_i \partial \theta_j} + \frac{1}{p_l^2(\boldsymbol{\theta})}\frac{\partial p_l(\boldsymbol{\theta})}{\partial \theta_i}\frac{\partial p_l(\boldsymbol{\theta})}{\partial \theta_j} \Bigg\} = \\
    \sum_l p_l(\boldsymbol{\theta}) \left[\frac{-1}{p_l(\boldsymbol{\theta})}\frac{\partial^2 p_l(\boldsymbol{\theta})}{\partial \theta_i \partial \theta_j} + \frac{1}{p_l^2(\boldsymbol{\theta})}\frac{\partial p_l(\boldsymbol{\theta})}{\partial \theta_i}\frac{\partial p_l(\boldsymbol{\theta})}{\partial \theta_j} \right] \\ \implies
    [\mathcal{F}_C]_{ij} = \sum_l \frac{1}{p_l(\boldsymbol{\theta})}\frac{\partial p_l(\boldsymbol{\theta})}{\partial \theta_i}\frac{\partial p_l(\boldsymbol{\theta})}{\partial \theta_j} 
\end{gathered}
\end{equation}
where we used the fact that
\begin{equation}
    \sum_l \frac{\partial^2 p_l(\boldsymbol{\theta})}{\partial \theta_i \partial \theta_j} =  \frac{\partial^2}{\partial \theta_i \partial \theta_j} \sum_l p_l(\boldsymbol{\theta}) = 0.
\end{equation}
As you can see, although the Fisher information is the matrix corresponding to second-order derivatives, after our analysis, it can be written as a product of first-order derivatives. This has an immediate advantage in the number of resources needed. It requires only the calculation of the first-order derivatives in the quantum computer. Then, the classical computer post-processes the $m$ first-order derivatives and stores the classical fisher information using $\mathcal{O}_C(m^2)$ classical memory.

\section{Parameterized Ansatz Families}
\label{appendix:parameterized_circuits}
The hardware-efficient ansatz families used in this manuscript can be visualized below. On the left side of Figure \ref{fig:quantum_circuits}, the circuits consist of a series of $R_y$-$R_x$ gates, followed by a series of controlled-NOT operations applied either in a nearest-neighbour fashion or with all-to-all connectivity. The total number of parameters for a $p$-layer $n$-qubit PQC is $m = 2np$. Different choices for the single-qubit gates and the entangling gates were tested without affecting the validity of our results.

For the second choice of the ansatz family, we applied the circuit visualized on the right side of Figure \ref{fig:quantum_circuits}. In this architecture, each qubit is first rotated by an $R_y$ gate. Then, for each layer, a series of controlled-$Z$ operations are applied in an all-to-all or nearest-neighbour fashion, and then a final series of rotation $R_y$ gates are applied in all qubits. The total number of parameters for a $p$-layer $n$-qubit PQC is $m = (p+1)n$.
\begin{figure}[!ht]
    \begin{center}
    \begin{tikzpicture}
        \node (img1) [scale=0.7]{
\begin{quantikz}
    \lstick{$\ket{0}$} &\gate{R_y(\theta_1)}\gategroup[6, steps=5, style={dashed, rounded corners,fill=blue!20, inner xsep=2pt}, background]{{\sc single layer}} &\gate{R_z(\theta_2)} &\ctrl{1} &\qw &\gate{X} & \meter{}\\
    \lstick{$\ket{0}$} &\gate{R_y(\theta_3)} &\gate{R_z(\theta_4)} & \gate{X} &\ctrl{1} & \qw & \meter{}\\
    \lstick{$\ket{0}$} &\gate{R_y(\theta_5)} &\gate{R_z(\theta_6)} &\ctrl{1} &\gate{X} &\qw & \meter{}\\
    \lstick{$\ket{0}$} &\gate{R_y(\theta_7)} &\gate{R_z(\theta_8)} &\gate{X} &\ctrl{1} & \qw & \meter{}\\
    \lstick{$\ket{0}$} &\gate{R_y(\theta_9)} &\gate{R_z(\theta_{10})} &\ctrl{1} &\gate{X} & \qw & \meter{}\\
    \lstick{$\ket{0}$} &\gate{R_y(\theta_{11})} &\gate{R_z(\theta_{12})} &\gate{X} &\qw &\ctrl{-5} & \meter{}
\end{quantikz}
        };
            \node[right=of img1, scale=0.7] {
\begin{quantikz}
    \lstick{$\ket{0}$} & \gate{R_y(\theta_{1})} & \ctrl{1} \gategroup[6, steps=8, style={dashed, rounded corners,fill=blue!20, inner xsep=2pt}, background]{{\sc single layer}} &\qw & \qw & \qw & \qw & \qw  &\gate{Z}  & \gate{R_y(\theta_{2})} &\meter{} \\
        \lstick{$\ket{0}$} & \gate{R_y(\theta_{3})} & \gate{Z} & \ctrl{1} & \qw & \qw & \qw & \qw   &\qw &\gate{R_y(\theta_{4})} & \meter{} \\
    \lstick{$\ket{0}$} & \gate{R_y(\theta_{5})} & \qw & \gate{Z} & \ctrl{1} & \qw & \qw & \qw &\qw & \gate{R_y(\theta_{6})} & \meter{} \\
    \lstick{$\ket{0}$} & \gate{R_y(\theta_{7})} & \qw & \qw & \gate{Z} & \ctrl{1} & \qw & \qw  &\qw & \gate{R_y(\theta_{8})} & \meter{} \\
    \lstick{$\ket{0}$} & \gate{R_y(\theta_{9})} & \qw & \qw & \qw & \gate{Z} & \ctrl{1} & \qw  &\qw & \gate{R_y(\theta_{10})} & \meter{} \\
    \lstick{$\ket{0}$} & \gate{R_y(\theta_{11})} & \qw & \qw & \qw & \qw & \gate{Z} & \qw &\ctrl{-5} & \gate{R_y(\theta_{12})} & \meter{}
\end{quantikz}
        };
    \end{tikzpicture}
    \caption{Six-qubit examples of hardware-efficient parameterized quantum circuits (with nearest-neighbours interaction) used for our experiments. The filled blue square corresponds to a single layer. Similar quantum circuits were used with an all-to-all connectivity.}
    \label{fig:quantum_circuits}
    \end{center}
\end{figure}
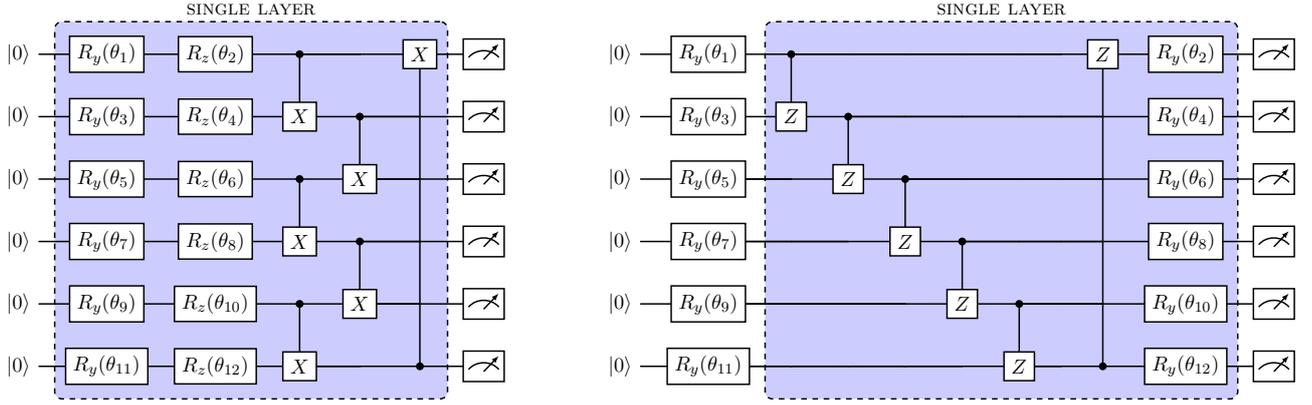

\section{Proof of Lemma \ref{lemma:probability}}
\label{appendix:proof_lemma}

In this section, we provide the proof of Lemma \ref{lemma:probability}. For clarity, we reintroduce our notation below. Let $S = [m]$, be the set of parameters that parameterize a quantum circuit. Let also $L \subseteq S$ with $|L| = l$ be the target subset of $l < m$ parameters that result in an independent change in the quantum state. Our goal is to quantify the probability of sampling a subset $S_k \subseteq S$ (of cardinality $|S_k| = l + k$) so that $L \subseteq S_k$.\\

Consider at first the case where $k = 0$. In that case, the probability that we sample the target subset can be calculated as:
\begin{equation}
    \Pr[ L = S_0] = \Bigg(\frac{l}{m}\Bigg) \Bigg(\frac{l - 1}{m-1}\Bigg) \ldots \Bigg(\frac{1}{m-l+1}\Bigg) = \frac{l!(m-l)!}{m!}
\end{equation}\\
where the right-hand side corresponds to the condition probability of sampling the first parameter at random and being in the target subset $\frac{l}{m}$, then sampling the second parameter at random and being also one of the remaining $(l-1)$ parameters in the target subset $\frac{l-1}{m}$ and so on until the last parameter that we sample to be also on the target subset $\frac{1}{m-l+1}$.
Then, consider $k = 1$. In that case, the probability of sampling the subset is calculated as:
\begin{equation}
\begin{gathered}
    \Pr[L \subseteq S_k] = \frac{m-l}{m}\frac{l}{m-1}\frac{l-1}{m-2} \ldots \frac{1}{m-l} + \frac{l}{m}\frac{m-l}{m-1}\frac{l-1}{m-2} \ldots \frac{1}{m-l} \\
    + \ldots + \frac{l}{m}\frac{l-1}{m-1}\ldots \frac{1}{m-l+1}\frac{m-l}{m} = (l+1) \Pr [S_0 = L]
\end{gathered}
\end{equation}
where the first term corresponds to the conditional probability that the first parameter that we sample is not in the target subset, but the rest are. Then, the second term corresponds to the first parameter being in the target subset, the second not being in the target subset, and the rest of the parameters again being in the target subset and so on for the remaining terms. In the exact same manner, we can calculate that for the general case $k = l+n$, the probability is:
\begin{equation}
    \Pr[ L \subseteq S_n] = \frac{(l + n)!}{n! l!}\frac{l!(m-l)!}{m!}.
\end{equation}

\section{Classical Fisher Information Matrix from Pauli Measurements}
\label{appendix:cfim_pauli_measurements}

In this section, we provide an illustrative example explaining why random Pauli measurements cannot be used in the update rule (Eq. \eqref{update_rule_rng}). In Figure \ref{fig:cfim_z_measurements} we can visualize how performing natural gradient with projective measurements on the $Z$-basis results in the inability of the classical optimization algorithm to converge. The results correspond to a Heisenberg model ($J=h=1$) of 9-qubits and a 3-layer parameterized quantum circuit seen on the left side of Figure \ref{fig:quantum_circuits}. In this case, the Natural Gradient, even with a small choice for the learning rate $\eta$, performs very large updates and is unable to converge. 

\begin{figure}
\begin{center}
\begin{tikzpicture}
\node (img)  {\includegraphics[scale=0.5]{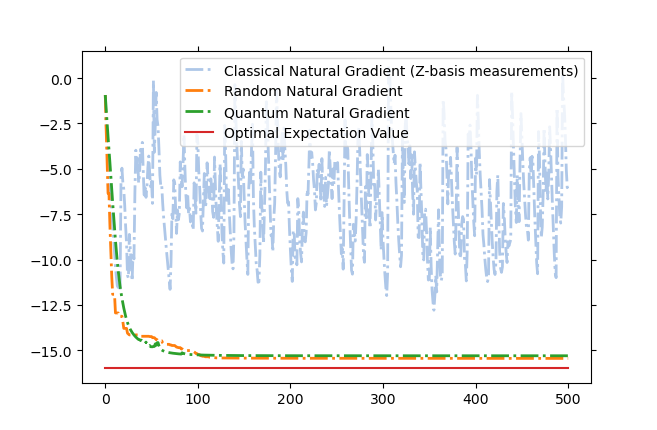}};
\node[below=of img, node distance=0cm, yshift=1.5cm] {\scriptsize Optimization Iterations};
\node[left=of img, node distance=0cm, xshift=0.7cm, rotate=90, anchor=center,yshift=-0.7cm] {\scriptsize $\mathcal{L}(\boldsymbol{\theta})$};
\end{tikzpicture}
\caption{Performance of RNG, QNG and Natural Gradient (with Pauli-$Z$ measurements) for a Heisenberg model of 9 qubits. While the RNG and QNG are able to converge fast (with the RNG) the natural gradient cannot converge as in most iterations it makes very large steps (even with a small learning rate).}
\label{fig:cfim_z_measurements}
\end{center}
\end{figure}

\section{Additional Experiments}
\label{sec:additional_experiments}

In this section, we provide some extra experiments, comparing our two proposed methods (see \emph{Random Natural Gradient} in Sec. \ref{sec:random_natural_gradient} and \emph{Stochastic-Coordinate Quantum Natural Gradient} in Sec. \ref{sec:stochastic_coordinate}) on the MaxCut problem (see Sec. \ref{sec:method_evaluation}). For our experiments, we employed the PQC visualized on the right side of Figure \ref{fig:quantum_circuits} with 3 layers. We sampled random 8 and 10-qubit unweighted 3-regular graph instances. 
\begin{figure*}[!ht]
\begin{tikzpicture}
\node (img1)  {\includegraphics[scale=0.54]{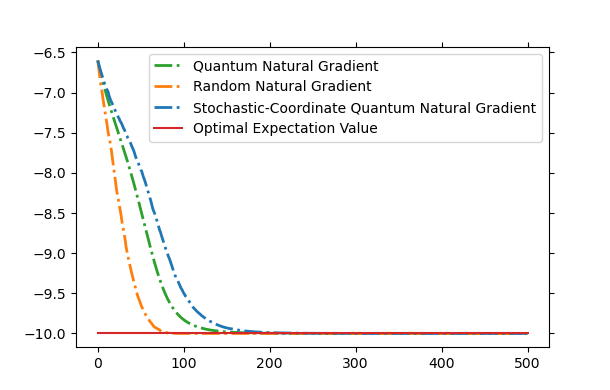}};
\node[below=of img1, node distance=0cm, yshift=1.3cm] (xlabel1) {\scriptsize Optimization Iterations};
\node[left=of img1, xshift=0.5cm, node distance=0cm, rotate=90, anchor=center,yshift=-0.7cm] {\scriptsize $\mathcal{L}(\boldsymbol{\theta})$};
\node[right=of img1, xshift=-1.6cm] (img2)  {\includegraphics[scale=0.54]{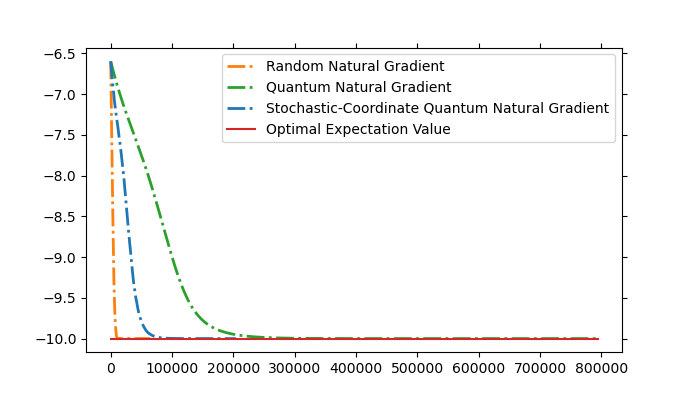}};
\node[below=of img2, node distance=0cm, yshift=1.3cm] (xlabel2) {\scriptsize Quantum State Preparations};
\node[left=of img2, xshift=0.5cm, node distance=0cm, rotate=90, anchor=center, yshift=-0.7cm] {\scriptsize $\mathcal{L}(\boldsymbol{\theta})$};
\node[below=of img1] (img3) {\includegraphics[scale=0.54]{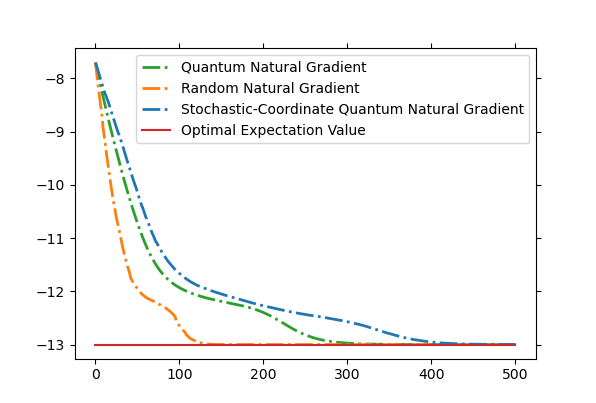}};
\node[below=of img3, node distance=0cm, yshift=1.3cm] (xlabel1) {\scriptsize Optimization Iterations};
\node[left=of img3, xshift=0.5cm, node distance=0cm, rotate=90, anchor=center,yshift=-0.7cm] {\scriptsize $\mathcal{L}(\boldsymbol{\theta})$};
\node[right=of img3, xshift=-1cm] (img4)  {\includegraphics[scale=0.54]{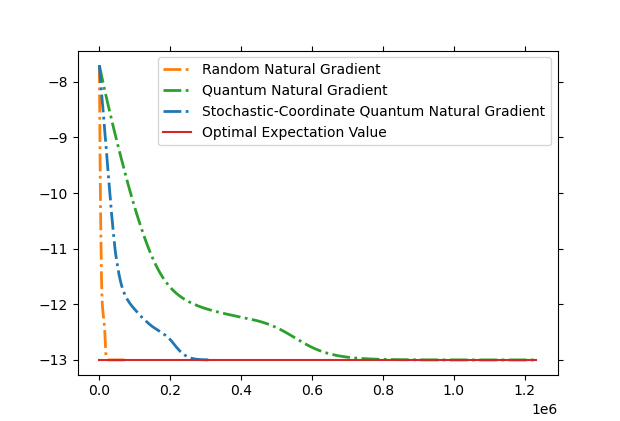}};
\node[below=of img4, node distance=0cm, yshift=1.3cm] (xlabel1) {\scriptsize Quantum State Preparations};
\node[left=of img4, xshift=0.5cm, node distance=0cm, rotate=90, anchor=center,yshift=-0.7cm] {\scriptsize $\mathcal{L}(\boldsymbol{\theta})$};
\end{tikzpicture}
\caption{Comparison of QNG, RNG and SC-QNG on MaxCut instances corresponding to 3-regular graphs of 8 qubits (up left and up right figures) and 10 qubits (bottom left and bottom right).}
\label{fig:optimization_methods_comparison}
\end{figure*}
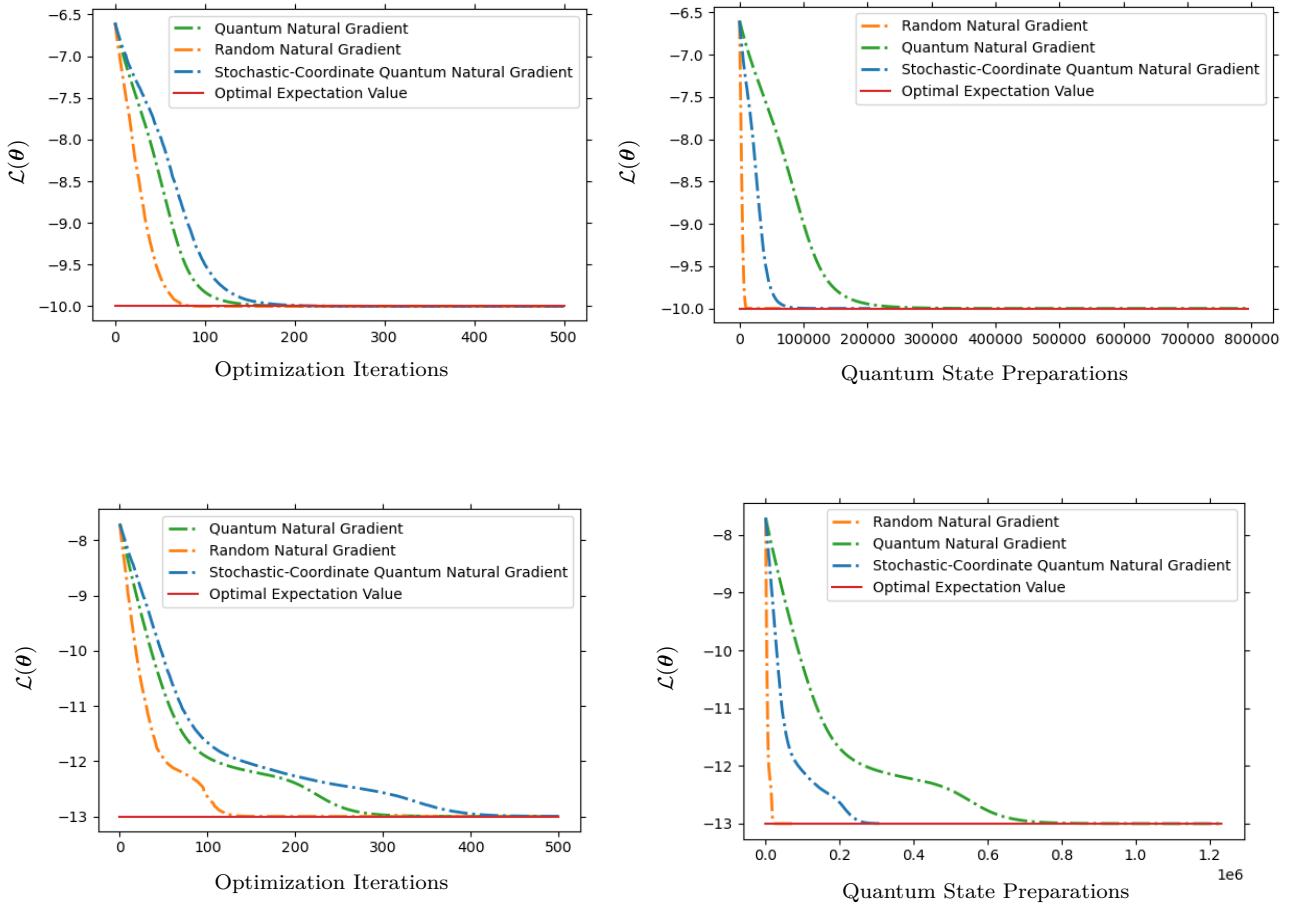

\section{Depth overhead in Information Matrices Calculation}
\label{appendix:depth_overhead}

As it was analyzed in Section \ref{sec:preliminaries}, the quantum natural gradient (QNG) requires the calculation of the QFIM at each iteration. Two different ways to measure the QFIM elements are explained below and are based on the \emph{parameter-shift} approach \cite{PhysRevA.103.012405}.

As we discussed in the preliminaries section (Sec. \ref{sec:preliminaries}), the elements of the quantum Fisher information matrix can be expressed as:
\begin{equation}
    [\mathcal{F}_Q(\boldsymbol{\theta)}]_{i,j} = -\frac{\partial^2}{\partial \theta_i\partial \theta_j} \Big{|} \bra{\psi(\boldsymbol{\theta'})}\ket{\psi(\boldsymbol{\theta})}\Big{|}^2 \Bigg{|}_{\boldsymbol{\theta'} = \boldsymbol{\theta}}
\end{equation}
which means that by using the parameter-shift rules (with $r=\frac{1}{2}$), the QFIM can be written as:
\begin{equation}
\begin{aligned}
     [\mathcal{F}_Q(\boldsymbol{\theta)}]_{i,j} = -\frac{1}{4} \Big[ |\bra{\psi(\boldsymbol{\theta})}\ket{\psi(\boldsymbol{\theta} + (\hat{\boldsymbol{e}}_i + \hat{\boldsymbol{e}}_j)\pi/2)}|^2 - |\bra{\psi(\boldsymbol{\theta})}\ket{\psi(\boldsymbol{\theta} + (\hat{\boldsymbol{e}}_i - \hat{\boldsymbol{e}}_j)\pi/2)}|^2 \\   -|\bra{\psi(\boldsymbol{\theta})}\ket{\psi(\boldsymbol{\theta} + (-\hat{\boldsymbol{e}}_i + \hat{\boldsymbol{e}}_j)\pi/2)}|^2 +|\bra{\psi(\boldsymbol{\theta})}\ket{\psi(\boldsymbol{\theta} - (\hat{\boldsymbol{e}}_i + \hat{\boldsymbol{e}}_j)\pi/2)}|^2 \Big]
\end{aligned}
\end{equation}
where $\hat{\boldsymbol{e}}_i$ is the unit vector pointing in the $i$-th direction. As we can see, for every matrix element of the QFIM, we have to calculate four expectation values. Focusing on a single expectation value out of the four, we can write it as:
\begin{equation}
|\bra{\psi(\boldsymbol{\theta})}\ket{\psi(\boldsymbol{\theta} + (\hat{\boldsymbol{e}}_i + \hat{\boldsymbol{e}}_j)\pi/2)} |^2 = |\bra{0} U^\dagger(\boldsymbol{\theta}) U(\boldsymbol{\theta} + (\hat{\boldsymbol{e}}_i + \hat{\boldsymbol{e}}_j)\pi/2)\ket{0}|^2
\end{equation}
which can be intuitively understood as the probability of measuring all qubits of the $U(\boldsymbol{\theta})U(\boldsymbol{\theta} + (\hat{\boldsymbol{e}}_i + \hat{\boldsymbol{e}}_j)\pi/2)\ket{0}$ state at $\ket{0}$. As such, for a parameterized quantum circuit of depth $d$ that generates a parameterized quantum state $\ket{\psi(\boldsymbol{\theta})}$, measuring its QFIM will require $O(m^2)$ quantum circuits of depth $2d$. On the other hand, one could use a SWAP test, as seen in the right side of Figure \ref{fig:qfim_circuits}, which reduces the depth of the circuit but requires twice the number of qubits.

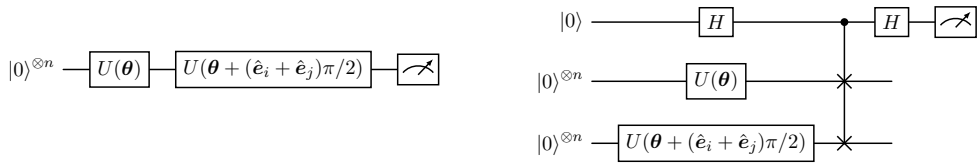
\begin{figure}[!ht]
    \begin{center}
    \begin{tikzpicture}
        \node (img1) [scale=0.7]{
            \begin{quantikz}
            \lstick{$\ket{0}^{\otimes n}$} &\gate[1]{U(\boldsymbol{\theta})} & \gate[1]{U(\boldsymbol{\theta} + (\hat{\boldsymbol{e}}_i + \hat{\boldsymbol{e}}_j)\pi/2)} &\meter{} \\
            \end{quantikz}
        };
            \node[right=of img1, scale=0.7] {
            \begin{quantikz}
            \lstick{$\ket{0}$} & \gate{H} &\ctrl{2} &\gate{H} &\meter{}\\
            \lstick{$\ket{0}^{\otimes n}$} & \gate{U(\boldsymbol{\theta})} &\targX & & \\
            \lstick{$\ket{0}^{\otimes n}$} & \gate[1]{U(\boldsymbol{\theta} + (\hat{\boldsymbol{e}}_i + \hat{\boldsymbol{e}}_j)\pi/2)} & \targX &&
            \end{quantikz}
        };
    \end{tikzpicture}
    \caption{Different ways to calculate overlaps that are required for the calculation of QFIM. On the left-hand side, we apply a circuit of depth $2d$ and measure the system of qubits. On the right side, we apply a SWAP test that requires twice the number of qubits but reduces the depth of the calculation.}
    \label{fig:qfim_circuits}
    \end{center}
\end{figure}

In our method, the resources needed to calculate the random CFIM are \emph{user-dependent}. The quantum state $\ket{\psi(\boldsymbol{\theta})}$ is prepared by a parameterized unitary $U(\boldsymbol{\theta})$ (of depth $d$) and the measurement is performed by rotating the state by a different unitary $V(\boldsymbol{\phi})$ which is usually chosen to have depth less than $d$ (see Figure \ref{fig:random_measurements}).

Specifically, the user first calculates the outcome probabilities $p_l(\boldsymbol{\theta})$ ($l \in \{0,1\}^n$) by employing the quantum circuit visualized on the left of Figure \ref{fig:random_measurements}. Then, they calculate the outcome probabilities $\frac{\partial p_l(\boldsymbol{\theta})}{\partial \theta_i}$ by using again the parameter-shift rule.
\begin{equation}
    \frac{\partial p_l(\boldsymbol{\theta}) }{\partial \theta_j} = \frac{1}{2}\Big( p_l(\boldsymbol{\theta} + \frac{\pi}{2}\hat{\boldsymbol{e}}_i) - p_l( \boldsymbol{\theta} - \frac{\pi}{2}\hat{\boldsymbol{e}}_i)\Big).
\end{equation}
One of the advantages is that the random-measurement unitary can be chosen to have depth less than that required to generate the quantum state $\ket{\psi(\boldsymbol{\theta})}$. For example, for a parameterized family of states such as those in Figure \ref{fig:quantum_circuits} consisting of $l$ layers, we always employed measurements that resulted in a depth-overhead less than that required for QNG. Specifically, we tested many different hardware-efficient ansatz families, such as those in Figure \ref{fig:quantum_circuits}, and our results remained consistent. That is, the hardware-efficient unitary consists of $k<l$ layers, showcasing that RNG remains more ``NISQy'' than QNG.

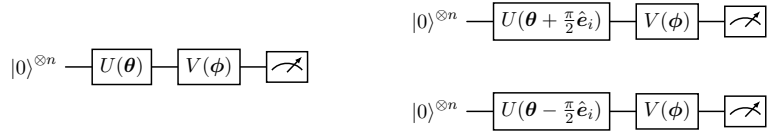
\begin{figure}
    \centering
    \begin{tikzpicture}
         \node (img1) [scale=0.7]{\begin{quantikz}
            \lstick{$\ket{0}^{\otimes n}$} &\gate[1]{U(\boldsymbol{\theta})} & \gate[1]{V(\boldsymbol{\phi})} &\meter{} 
        \end{quantikz}};
            \node[right=of img1, scale=0.7] {
            \begin{quantikz}
            \lstick{$\ket{0}^{\otimes n}$} &\gate[1]{U(\boldsymbol{\theta} + \frac{\pi}{2}\hat{\boldsymbol{e}}_i)} & \gate[1]{V(\boldsymbol{\phi})} &\meter{}\\
            \\
            \lstick{$\ket{0}^{\otimes n}$} &\gate[1]{U(\boldsymbol{\theta} - \frac{\pi}{2}\hat{\boldsymbol{e}}_i)} & \gate[1]{V(\boldsymbol{\phi})} &\meter{}
            \end{quantikz}
        };
    \end{tikzpicture}
    \caption{Quantum circuit used for the calculation of both $p_l(\boldsymbol{\theta})$ and $ \frac{\partial p_l(\boldsymbol{\theta}) }{\partial \theta_j}$ with $l \in \{0,1\}^n$ at different bases. First, the user prepares the parameterized quantum state $\ket{\psi(\boldsymbol{\theta})}$ and then, the quantum state is measured on a random basis. The random basis measurement is performed by rotating the parameterized state by a unitary $V(\boldsymbol{\phi})$ and then measured in the computational basis.}
    \label{fig:random_measurements}
\end{figure}

\end{document}